\documentclass[aps,preprintnumbers,nofootinbib,twocolumn,longbibliography]{revtex4-1}
\newcommand{\pagenumbaa}{1}

\usepackage[usenames,dvipsnames]{xcolor}
\usepackage[latin1]{inputenc}
\usepackage[T1]{fontenc} 
\usepackage{graphicx} 
\usepackage[tight]{subfigure}
\usepackage{amssymb}
\usepackage{amsfonts}
\usepackage{amsmath}
\usepackage{amsthm}
\usepackage{appendix}
\usepackage{booktabs} 
\usepackage[rightcaption]{sidecap}
\usepackage{bbm} 
\usepackage{amsthm}
\usepackage{hyperref}
\usepackage{graphicx}
\usepackage{enumerate}
\usepackage{tikz}
\usetikzlibrary{decorations.pathreplacing,decorations.pathmorphing}
\usepackage{color}
\usepackage{multirow}
\usepackage{listings}

\theoremstyle{plain} \newtheorem{lemma}{Lemma}
\theoremstyle{plain} \newtheorem{proposition}{Proposition}
\theoremstyle{definition} \newtheorem{definition}{Definition}
\theoremstyle{plain} \newtheorem{example}{Example}

\newcommand*{\id}{\mathrm{id}}
\newcommand*{\tr}{\mathrm{tr}}
\newcommand*{\ket}[1]{| #1 \rangle}

\newcommand{\proj}[1]{|#1\rangle\!\langle #1|}

\begin{document}

\title{True randomness from realistic quantum devices}
\author{Daniela Frauchiger}
\author{Renato Renner}
\author{Matthias Troyer}
\affiliation
{Institute for  Theoretical Physics, ETH Zurich, Switzerland}

\begin{abstract}
  Even if the output of a Random Number Generator (RNG) is perfectly uniformly distributed, it may be correlated to pre-existing information and therefore be predictable. Statistical tests are thus not sufficient to guarantee that an RNG is usable for applications, e.g., in cryptography or gambling, where unpredictability is important. To enable such applications a stronger notion of randomness, termed ``true randomness'', is required, which includes independence from prior information.

  Quantum systems are particularly suitable for true randomness generation, as their unpredictability can be proved based on physical principles. Practical implementations of Quantum RNGs (QRNGs) are however always subject to noise, i.e., influences which are not fully controlled. This reduces the quality of the raw randomness generated by the device, making it necessary to post-process it. Here we provide a framework to analyse realistic QRNGs and to determine the post-processing that is necessary to turn their raw output into true randomness. 
\end{abstract}

\maketitle

\setcounter{page}{\pagenumbaa}
\thispagestyle{plain}

\section{Introduction} \label{sec:introduction}

The generation of good random numbers is not only of academic interest but also very relevant for practice. The quality of the randomness used by cryptographic systems, for instance, is essential to guarantee their security. However, even among manufacturers of \emph{Random Number Generators (RNGs)}, there does not seem to exist a consensus about how to define randomness or measure its quality. What is worse, many of the RNGs used in practice are clearly insufficient.  Recently, an analysis of cryptographic public keys available on the web revealed that they were created from very weak randomness, a fact that can be exploited to get hold of a significant number of the associated private keys~\cite{epfl12}. The quality of RNGs is also important for non-cryptographic applications, e.g., in data analysis, numerical simulations~\cite{metropolis49} (where bad randomness affects the reliability of the results), or  gambling.

\subsection{What is true randomness?}

Several conceptually quite diverse approaches to randomness have been considered in the literature. One possibility is to view randomness as a property of actual values. Specifically, a bit sequence may be called random if its Kolmogorov complexity is maximal.\footnote{The \emph{Kolmogorov complexity} of a bit string corresponds to the length of the shortest program that reproduces the string. Whereas for finite bit strings the notion depends on the choice of the language used to describe the program, this dependence disappears asymptotically for long strings.} However, this approach only makes sense asymptotically for infinitely long sequences of bits. Furthermore, the Kolmogorov complexity is not computable~\cite{kolmogorov}. But, even more importantly, the Kolmogorov complexity of a bit sequence does not tell us anything about how well the bits can be predicted, thus precluding their use for any application where unpredictability is relevant, such as the drawing of lottery numbers.\footnote{Last week's lottery numbers will probably have the same Kolmogorov complexity as next week's numbers. Nevertheless, we would not reuse them.}

The same problem arises for the commonly used statistical criterion of randomness, which demands uniform distribution, i.e., that each value is equally likely.  To see this, imagine two RNGs that output the same, previously stored, bit string $K$. If $K$ is uniformly distributed then both RNGs are likely to pass any statistical test of uniformity. However, given access to one of the RNGs, the output of the other can be predicted (as it is identical by construction).  Hence, even if the quality of an RNG is certified by statistical tests, its use for drawing lottery numbers, for instance, may be prohibited. 

Here we take a different approach where, instead of the actual numbers or their statistics, we consider the \emph{process} that generates them.  We call a process \emph{truly random} if its outcome is uniformly distributed \emph{and} independent of all information available in advance. A formal definition can be found in Section~\ref{trueRandomness}. We stress that this definition guarantees \emph{unpredictability} and thus overcomes the problems of the weaker notions of randomness described above. 

\subsection{How to generate true randomness?}

Having now specified what we mean by true randomness, we can turn to the question of how to generate truly random numbers. Let us first note that \emph{Pseudo-Random Number Generators (PRNGs)}, which are widely used in practice,\footnote{PRNGs generate long sequences of random-looking numbers by applying a function to a short \emph{seed} of random bits. They  are convenient because they do not require any specific hardware and because they are efficient.} cannot meet our criterion for true randomness. Still,  the output of a PRNG may be computationally indistinguishable from truly random numbers --- a property that is sufficient for many applications. But to achieve this guarantee, the PRNG must be initialised with a seed that is itself truly random.  Hence, a PRNG alone can never replace an RNG. 

\emph{Hardware-based RNGs} make use of a physical process to generate randomness. Because classical physics is deterministic, RNGs that rely on phenomena described within a purely classical noise model, such as thermal noise~\cite{zhun01}, can only be proved random under the assumption that the microscopic details of the system are inaccessible. This assumption is usually hard to justify physically. For example, processes in resistors and Zener diodes have memory effects~\cite{stipcevic11}. Hence, someone who is able to gather information about the microscopic state of the device --- or even influence it\footnote{For example, the state of the device may be influenced by voltage changes of its power supply or by incident radiation.} --- could predict its future behaviour.

In contrast to this, measurements on quantum systems are intrinsically probabilistic. It is therefore possible to prove, based on physical principles, that the output of a \emph{Quantum Random Number Generator (QRNG)} is truly random.  Moreover, a recent result \cite{colbeck11} implies that the unpredictability does not depend on any completeness assumptions about quantum theory, i.e., unpredictability is guaranteed within all possible extensions of the theory.\footnote{More precisely, consider an arbitrary alternative theory that is compatible with quantum theory (in the sense that its predictions do not contradict quantum theory) and in which true  randomness can exist in principle. Then the outcome of a process is unpredictable within the alternative theory whenever it is unpredictable within quantum theory.}

\subsection{Randomness from imperfect devices} \label{sec:introex}

In \emph{practical implementations} of QRNGs, the desired quantum process can never be realised perfectly. There are always influences that are not fully controlled by the manufacturer of the device. In the following, we generally term such influences ``noise''. While noise cannot be controlled, we can also not be sure that it is actually random.  This affects the guarantees we can provide about the randomness of the output of the QRNG.

\begin{figure}
\begin{tikzpicture}
\scriptsize
\draw[->,decorate,decoration={snake,post length=1 mm,amplitude=1mm,segment length=6mm}] (1,0)--(3.8,0);
\draw[line width=0.4mm] (3.75,-0.5) --(4.25,0.5);
\draw[->,decorate,decoration={snake,post length=1 mm,amplitude=1mm,segment length=6mm}] (4,0.2)--(4,2);
\draw[->,decorate,decoration={snake,post length=1 mm,amplitude=1mm,segment length=6mm}] (4.2,0)--(6,0);
\draw (4.5,3) arc(0:180:0.5);
\draw (4.5,3)--(4.5,2.2)--(3.5,2.2)--(3.5,3);
\draw (7,-0.5) arc(-90:90:0.5);
\draw (7,-0.5)--(6.2,-0.5)--(6.2,0.5)--(7,0.5);
\node at (4,2.8) {{\normalsize $D_v$}};
\node at (6.8,0)  {{\normalsize $D_h$}};
\node at (6.5,3)  {{\normalsize $X=(X_v,X_h)$}};
\node at (7.5,2)  {{\normalsize $X_{v,h}=\left\{\begin{array}{ll}1&\text{if } D_{v,h} \text{ clicks.}\\0&\text{else}\end{array}\right.$}};
\end{tikzpicture}
\caption{\emph{QRNG based on a Polarising Beam Splitter (PBS).} The beam splitter reflects or transmits incoming light depending on its polarisation (vertical or horizontal). In the ideal case, where the incoming light consists of diagonally polarised single-photon pulses and where the detectors click with certainty if and only if a photon arrives, each of the two outcomes $X=(0,1)$ and $X=(1,0)$ occurs with probability $\frac{1}{2}$.}\label{fig:PBS1} 
\vskip -0.5cm
\end{figure}
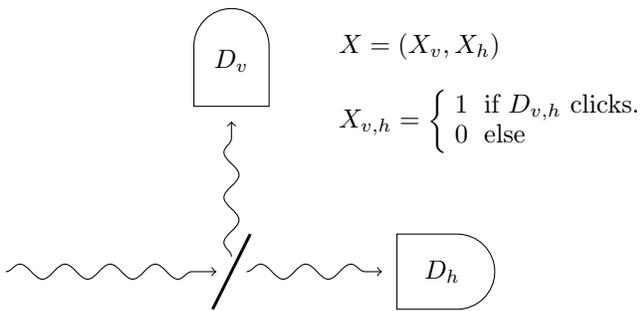

As an example, consider a QRNG based on a Polarising Beam Splitter (PBS), which reflects vertically polarised photons and transmits horizontally polarised photons  as illustrated in Fig.~\ref{fig:PBS1}  (see also Fig.~\ref{QRNG} as well as the more detailed description in Section~\ref{example}). To generate randomness, the PBS is illuminated by a diagonally polarised light pulse. After passing through the PBS the light hits one of two detectors, labeled $D_v$ and $D_h$, depending on whether it was reflected ($v$) or transmitted ($h$). The output of the device is $X = (X_v, X_h)$, where $X_{v,h}$ are bits indicating whether the corresponding detector $D_{v,h}$ clicked. In the ideal case, where the light pulses contain exactly one photon and where the  detectors are maximally efficient, only the outcomes $X = (0,1)$ and $X = (1,0)$ are possible. The process thus corresponds to a polarisation measurement of a diagonally polarised photon with respect to the horizontal and vertical direction. According to quantum theory, the resulting bit, indicating whether $X=(0,1)$ or $X=(1,0)$, is uniformly distributed and unpredictable. That is, it is truly random. 

The situation changes if the ideal detectors are replaced by imperfect ones, which sometimes fail to notice an incoming photon, and if the light source sometimes emits pulses with more than one photon. Consider the case where the pulse is so strong that   (with high probability) there are photons hitting both $D_v$ and $D_h$ at the same time.  We now still obtain outcomes  $X=(0,1)$ or $X=(1,0)$ since one of the detectors may not click. But the outcome can no longer be interpreted as the result of a polarisation measurement. Rather, it is determined by the detectors' probabilistic behaviour, i.e., whether they were sensitive at the moment when the light pulses arrived. In other words, the  device outputs detector noise instead of quantum randomness originating from the PBS! 

The example illustrates that the output of an imperfect QRNG may be correlated to noise and, hence, potentially to the history of the device or its  environment. It is therefore no longer guaranteed to be truly random. But, luckily, this can be fixed. By appropriate post-processing of the \emph{raw randomness} generated by an imperfect device it is still possible to obtain true randomness.

\subsection{Turning noisy into true randomness}

A main contribution of this work is to provide a framework for analysing practical (and, hence, imperfect) implementations of QRNGs and determining the post-processing that is necessary to turn their raw outputs into truly random numbers. The framework is general; it is applicable to any QRNG that can be modelled within quantum theory. While the derived guarantees on the randomness of a QRNG thus rely on its correct modelling, as we shall see, no additional completeness assumptions need to be made.

\begin{figure}[h]
\begin{tikzpicture}[scale=0.3]
\scriptsize
\draw (0,0.5)--(10,0.5);
\draw (0,-0.5)--(10,-0.5);
\foreach \x in {0,...,10} \draw (\x,-0.5) -- (\x,0.5);
\draw[->] (10.5,0)--(11.5,0);
\draw (11.75,1.5)--(11.75,-1.5)--(15.25,-1.5)--(15.25,1.5)--(11.75,1.5);
\node at (13.5,0.75) {post-};
\node at (13.5,-0.75) {process.};
\draw[->] (15.5,0)--(16.5,0);
\draw (17,0.5)--(24,0.5);
\draw (17,-0.5)--(24,-0.5);
\foreach \x in {17,...,24} \draw (\x,-0.5) -- (\x,0.5);
\node at (5,-1.5) {$X$};
\draw[decorate,decoration={snake,post length=1 mm,amplitude=1mm,segment length=6mm}] (5.5,-1.5)--(13,-4);
\node at (13.5,-4) {$W$};
\node at (20.5,-1.5) {$f(X)$};
\draw[decorate,decoration={snake,post length=1 mm,amplitude=1mm,segment length=6mm}] (19.25,-1.5)--(14,-4);
\draw[line width=0.4mm] (17,-2.25)--(17.5,-3.25);
\draw[line width=0.4mm] (17,-3.25)--(17.5,-2.25);
\draw[<->] (0,1.25)--(10,1.25);
\node at (5,1.75) {$n$};
\draw[<->] (17,1.25)--(24,1.25);
\node at (20.5,1.75) {$\ell$};
\end{tikzpicture}
\caption{\emph{Post-processing by block-wise hashing.} The raw randomness may depend on side information $W$. To post-process it, the randomness is casted into blocks consisting of $n$ bits. Each block, $X$, is given as input to a hash function, $f$, that outputs a shorter block, $f(X)$, of $\ell$ bits. The framework we propose allows to determine the fraction $\ell/n$ such that $f(X)$ is truly random, i.e., independent of $W$.}\label{fig:postProc} 
\end{figure}
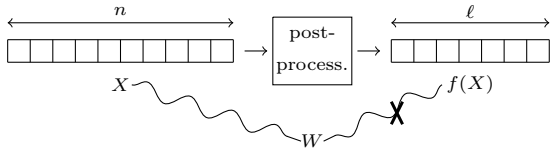

For the post-processing we consider block-wise hashing: the raw randomness is casted into $n$-bit strings, to which a \emph{hash function} is applied that outputs $\ell$-bit strings (where $\ell$ is generally smaller than $n$) as illustrated in Fig.~\ref{fig:postProc}. If the parameters $n$ and $\ell$ as well as the hash function are well chosen then the final $\ell$-bit strings are, to good approximation, truly random. Note that this is also known as \emph{randomness extraction} or \emph{privacy amplification}, and has been studied extensively in the context of classical data processing and cryptography~\cite{BBR88,ImLeLu89,ImpZuc89,BBCM95,NisZuc96,Trevisan01,Shaltiel02}. However, only few hashing techniques are known to be sound in the context of quantum information~\cite{RenKoe06,Renner06,KonTer08,TaShma09,TSSR10,BenTa12,DPVR12}. In this work we focus on a particular hashing procedure, called \emph{two-universal hashing} (or \emph{leftover hashing}), for which this property has been proved~\cite{Renner06,TSSR10} and which, additionally, is computationally very efficient and therefore suitable for practical purposes~\cite{CarWeg79,WegCar81,Stinson02}. As the name suggests, these hash functions are universal in the sense that their use does not depend on the details of the raw randomness, but only on its overall quality, which is measured in terms of an entropic quantity, called \emph{min-entropy} (see Section~\ref{sec:hash} below). An explicit implementation of such a hash function is described in Appendix~\ref{sec:impl}.

% The only relevant parameters are the input length and the output length, $n$ and $\ellk$, of the function. These parameters must be chosen depending on the quality of the raw randomness. The more ``good'' randomness it contains, the larger is the \emph{compression rate}, i.e., the ratio between $\ell$ and $n$. 

\subsection{Related work}

To analyse the quality of the raw randomness generated by a realistic device we must take into account any noise as side information. This is important as our goal is to prove unpredictability of the final randomness without assuming that the noise is itself unpredictable.  Somewhat surprisingly, such side information is usually not considered in the literature on RNGs. The only exception, to our knowledge, is the work of Gabriel \emph{et al.}~\cite{gabriel10} and Ma \emph{et al.}~\cite{ma12}, but the post-processing procedure applied in these cases does not guarantee the desired independence from noise either. The former study~\cite{gabriel10} uses Shannon entropy to quantify the randomness produced by the device, which only gives an upper bound on the amount of extractable independent randomness (see Appendix~\ref{vN}). In the framework proposed by Ma \emph{et al.}~\cite{ma12} the min-entropy is used to quantify randomness in the presence of noise, yet without conditioning on the noise. The final randomness is then uniformly distributed, but there is no guarantee that it is noise-independent.

We also note that randomness generation as studied in this work is different from \emph{device-independent randomness expansion}~\cite{colbeck09, pironio10}. In the latter, randomness is generated and certified based on correlations from local measurements on entangled quantum systems, whereas no assumptions about the internal workings of the device that generates the correlations are necessary. However, as in the case of pseudo-random number generation, this requires a source of initial randomness, so that an  RNG would still be needed to run the scheme. 

From a practical perspective, another difference between our approach and device-independent randomness expansion is that implementations of the latter are very challenging with  state-of-the-art technology and their efficiency is low (we know of only $42$ numbers generated this way~\cite{pironio10}). In contrast, QRNGs are easier to implement (commercial devices are already available) and efficient (the bit rates are of the order of Mbits/s~\cite{whitepaper}).

\bigskip

The remainder of this paper is organised as follows. In Section~\ref{trueRandomness} we formally define true randomness. Section~\ref{sec:hash} reviews the concept of hashing (cf.\ Appendix~\ref{sec:impl} for an implementation of hashing). In Section~\ref{sec:QRNGmodel} we explain how to generally model and analyse (imperfect) implementations of QRNGs and, in particular, introduce the notion of ``classical noise'', which we use to assess the quality of the raw randomness generated by such a device.  Furthermore, we show that the approach is \emph{complete}, in the sense that our statements about randomness remain valid if the model was replaced by another compatible model.  In Section~\ref{example} we describe a simple example that illustrates how the framework can be applied (see also Appendix~\ref{quantis} for a more realistic example). 

\section{Preliminaries}

\subsection{True randomness}\label{trueRandomness}

To introduce a formal and quantitative definition of true randomness, we make use of the notion of \emph{space time variables}. These are random variables with an associated coordinate that indicates the physical location of the value in relativistic space time~\cite{colbeck11}.\footnote{If a value is accessible at multiple places and times, this may be modelled by a set of space time variables with the same value but different space time coordinates.} We model the output, $X$, of a random process as well as all \emph{side information}, i.e., any additional variables that may be correlated to $X$, by space time variables. The space time coordinate of $X$ should be interpreted as the event where the process generating $X$ is started. For side information, the coordinates of the corresponding space time variables indicate when and where this information is accessible.

\begin{definition} \label{def_random} $X$ is called \emph{$\epsilon$-truly random} if it is $\epsilon$-close to uniform and uncorrelated to all other space time variables which are not in the future light cone of $X$. Denoting this set by $\Gamma_X$, this can be expressed as 
\begin{align}\label{eq:rand} \frac{1}{2} \| P_{X\Gamma_X}-P_{\bar{X}}\times P_{\Gamma_X} \|_1 \leq \epsilon 
\end{align}
where 
\begin{align} \label{eq:PX}
 P_{\bar{X}}(x)=\frac{1}{|\mathcal{X}|}\hspace{1cm}\forall\ x, 
\end{align}
and where $\frac{1}{2} \|P_X-Q_X \|_1 = \frac{1}{2} \sum_x |P_X(x)-Q_X(x)|$ is the trace distance.  
\end{definition}

We remark that the trace distance has the following operational interpretation: If two probability distributions are $\epsilon$-close to each other in trace distance then one may consider the two scenarios described by them as identical except with probability at most $\epsilon$~\cite{RenKoe06}. We also note that Definition~\ref{def_random} is closely related to the concept of a free choice~\cite{colbeck11,colbeck12,colbeck13}.

\subsection{Leftover hashing with side information}\label{sec:hash}

As described in the introduction, the post-processing of the raw randomness generated by a device consists of applying a hash function. The \emph{Leftover Hash Lemma with Side Information}, which we state below, tells us how to choose the parameters of the hash function depending on the quality of the raw randomness. To quantify the latter, we need the notion of min-entropy, which we now briefly review. 

Let $X$ be the value whose randomness we would like to quantify and let $W$ be any other random variable that models side information about $X$. Formally, this is described by a joint probability distribution $P_{X W}$. The \emph{conditional min-entropy of $X$ given $W$}, denoted $H_{\min}(X|W)$, corresponds to the probability of guessing $X$ given $W$. It is defined by
\begin{align}
2^{-H_{\min}(X|W)}=\sum\limits_{w}P_W(w)2^{-H_{\min}(X|W=w)} 
\label{eq:condHmin1}
\end{align}
where
\begin{align}
H_{\min}(X|W=w)=-\log_2\left[\max\limits_{x}P_{X|W=w}(x|w)\right].
\label{eq:Hmin}
\end{align}
All logarithms are with respect to base $2$.

Since we are studying quantum devices, we will also need to consider the more general case of non-classical side information, i.e., $X$ may be correlated to a quantum system, which we denote by $E$. Let $\rho_E^x$ be the state of $E$ when $X = x$. This situation can be characterised conveniently by a \emph{cq-state},
\begin{align}\label{CQstate} 
  \rho_{XE}=\sum_{x\in\mathcal{X}}P_X(x)|x\rangle\langle x|\otimes\rho_E^x \ ,
\end{align}
 where one thinks of the classical value $x \in \mathcal{X}$ as encoded in mutually orthogonal states $\{|x\rangle\}_{x\in\mathcal{X}}$ on a quantum system $X$.  The conditional min-entropy of $X$ given $E$ is then defined as 
\begin{multline} \label{HminQ}
  H_{\min}(X|E) \\= \sup \{\lambda : \, 2^{-\lambda} \id_X\otimes\sigma_E-\rho_{XE} \geq 0; \sigma_E \geq 0\}   \ .
\end{multline}
It has been shown that this corresponds to the maximum probability of guessing $X$ given $E$, and therefore naturally generalises the classical conditional min-entropy defined above~\cite{koenig09}. 

For later use, we also note that the min-entropy satisfies the \emph{data processing inequality}~\cite{Renner06,Beaudry12}. One way to state this is that discarding side information can only increase the entropy, i.e., for any two systems $E$ and $E'$ we have
\begin{align} \label{eq:DPI}
  H_{\min}(X|E E') \leq H_{\min}(X|E) \ .
\end{align}

To formulate the Leftover Hash Lemma, we will also employ a quantum version of the independence condition occurring in~\eqref{eq:rand},  
\begin{align*}
  \frac{1}{2} \bigl\| \rho_{XE}-\rho_{\bar{X}}\otimes\rho_{E} \bigr\|_1 \leq \epsilon \ ,
\end{align*}
where $\|\cdot\|_1=\tr(|\cdot|)$ denotes the trace norm and where $\rho_{\bar{X}}=\frac{1}{|\mathcal{X}|}\id_X$ is the fully mixed density operator on $X$. The condition characterises the states for which $X$ can be considered (almost) uniformly distributed and independent of $E$. An important property of this condition is that the trace norm can only decrease if we apply a physical mapping, e.g., a measurement, on the system~$E$. 

Leftover hashing is a special case of randomness extraction (see the introduction) where the hash function is chosen from a particular class of functions,  called \emph{two-universal}~\cite{CarWeg79,WegCar81, BBR88,ImLeLu89,BBCM95, Stinson02}.
They are defined as families  $\mathcal{F}$ of functions from $\mathcal{X}$ to $\{0,1\}^\ell$ such that
\[\Pr(f(x)=f(x'))\leq \frac{1}{2^\ell}\]
for any distinct $x,x'\in\mathcal{X}$ and $f$ chosen uniformly at random from $\mathcal{F}$.

We now have all ingredients ready to state the Leftover Hash Lemma with Side Information.

\begin{lemma}\label{lhl}
Let $\rho_{XE}$ be a cq-state and let $\mathcal{F}$ be a two-universal family of hash functions from $\mathcal{X}$ to $\{0,1\}^\ell$. Then
\begin{align*}
\frac{1}{2}\bigl\|\rho_{F(X)EF}-\rho_{\bar{Z}}\otimes\rho_{EF}\bigr\|_1 \leq 2^{-\frac{1}{2}(H_{\min}(X|E)-\ell)}:=\epsilon_{\text{hash}} \ , 
\end{align*}
where 
\begin{align*}
\rho_{F(X)EF}=\sum_{f\in\mathcal{F}} \frac{1}{|\mathcal{F}|} \rho_{f(X)E}\otimes|f\rangle\langle f| \ .
\end{align*}
and where $\rho_{\bar{Z}}$ is the fully mixed density operator on the space encoding $\{0,1\}^\ell$. 
\end{lemma}

The lemma tells us that, whenever $\ell<H_{\min}(X|E)$, the output $f(X)$ of the hash function is uniform and independent of $E$, except with probability $\epsilon_{\text{hash}} < 1$. The entropy $H_{\min}(X|E)$ thus corresponds to the amount of randomness that can be extracted from $X$ if one requires uniformity and independence from $E$. Furthermore, the deviation $\epsilon_{\text{hash}}$ decreases exponentially as $H_{\min}(X|E)$ increases. The above holds on average, for $f$ chosen uniformly at random from the family $\mathcal{F}$.  Note that the inclusion of $f$ in the state is important as this ensures that $f(X)$ is  random even if the function $f$ is known.

We also remark that the version of the Leftover Hash Lemma given above, while sufficient for our purposes, could be made almost tight by replacing the min-entropy by the smooth min-entropy~\cite{Renner06,TSSR10}. For the case of classical side information, $W$, one may also use the Shannon entropy as an upper bound. More precisely, for any function $f: \mathcal{X}\rightarrow\{0,1\}^\ell$ such that $\|P_{f(X)W}- {P_{\bar{Z}} \times P_{W}} \|_1\leq\epsilon$ (where $P_{\bar{Z}}$ is the uniform distribution on $ \{0,1\}^\ell$, cf.~\eqref{eq:PX}) it holds that
\begin{equation}
\ell \leq H(X|W)+4\epsilon\log \ell+2h(\epsilon),
\end{equation}
where $h(\cdot)$ is the binary entropy function  (cf.\ Appendix~\ref{vN} for a proof).  

For a device that generates a continuous sequence of output bits, the hash function is usually applied block-wise and the outcomes are concatenated. In this case, as shown in Appendix~\ref{sec:seed}, it is sufficient to choose the hash function once, using randomness that is independent of the device. The same function can then be reused for all blocks. For practice, this means that the hash function may be selected already when manufacturing the device (using some independent randomness) and be hardcoded on the device. An efficient implementation of two-universal hashing is presented in Appendix \ref{sec:impl}.

\section{A Framework for True Quantum Randomness Generation}\label{sec:QRNGmodel}

\subsection{General idea}

On an abstract level, a QRNG may be modelled as a process where a quantum system is prepared in a fixed state and then measured. Under the assumption that (i)~the state of the system is pure and that (ii)~the measurement on the system is projective the outcomes are truly random, i.e., independent of anything preexisting~\cite{colbeck11}. However, for realistic implementations, neither of the two assumptions is usually satisfied. If the preparation is noisy then the system is, prior to the measurement, generally in a mixed state. Furthermore, an imperfect implementation of a projective measurement, e.g., with inefficient detectors, is no longer projective, but rather acts like a  general Positive-Operator Valued Measure (POVM) on the system~\cite{dariano05}. These deviations from the assumptions~(i) and~(ii) mean that there exists side information that may be correlated to the outcomes of the measurement. (For example, for a mixed state, the side information could indicate which component of the mixture was prepared.) Our task is therefore to quantify the amount of independent randomness that is still present in the measurement outcomes. More precisely, we need to find a lower bound on the conditional min-entropy of the measurement outcomes given the side information. This then corresponds to the number of truly random bits we can extract by two-universal hashing.

 \subsection{Side information about realistic QRNGs}

 It follows from Neumark's theorem that a POVM can always be seen as a projective measurement on a product space, consisting of the original space on which the POVM is defined and an additional space (this is known as a \emph{Neumark extension} of the POVM \cite{Peres90}). We can therefore, even for a noisy QRNG, assume without loss of generality that the measurement is projective, but possibly on a larger space, which includes additional degrees of freedom not under our control.\footnote{Technically, the lack of control of certain degrees of freedom is modelled by a mixed initial state of the corresponding subspace.} These additional degrees of freedom will in general be correlated to side information. Once the projective measurement is known, it is not necessary to model the interaction of the QRNG with the environment. As we shall see, any side information will be taken into account automatically in our framework.

Let us stress, however, that given a description of a QRNG in terms of a general POVM, the choice of a Neumark extension is not unique and may impact the analysis.  In particular, the min-entropy, and hence the amount of extractable randomness, can be different for different extensions (see Appendix~\ref{POVM}). Therefore, it is necessary that the physical model of the QRNG specifies the projective measurement explicitly. 

\begin{definition}\label{def:QRNG}
A \emph{QRNG} is defined by a density operator $\rho_{S}$ on a system  $S$ together with  a projective measurement\footnote{Mathematically, a projective measurement on $S$ is defined by a family of projectors, $\Pi_S^x$, such that $\sum_{x} \Pi_S^x = \id_S$.} $\{\Pi_S^x\}_{x\in\mathcal{X}}$ on $S$. The \emph{raw randomness} is the random variable $X$ obtained by applying this measurement to a system prepared according to $\rho_S$. 
\end{definition}

Note that the probability distribution $P_X$ of $X$ is therefore given by the Born rule

\begin{align} \label{eq:Born}
P_X(x)=\tr(\Pi_S^x\rho_S) \ .
\end{align}

\begin{example}[Ideal PBS-based QRNG] \label{ex:idealPBS}
  Consider a QRNG based on a Polarising Beam Splitter (PBS), as described in the introduction (see Section~\ref{sec:introex}). One may view the source as well as the PBS as part of the state preparation, so that $\rho_S$ corresponds to the joint state of the two light modes traveling to the detectors, $D_v$ and $D_h$, respectively (see Fig.~\ref{fig:PBS1}). 
 In the ideal case, where the source emits one single diagonally polarised photon, $\rho_S = \proj{\phi}_{D_v D_h}$ is the pure state defined by 
\begin{align*} 
  \ket{\phi}_{D_v D_h} = {\textstyle \frac{1}{\sqrt{2}}} \bigl( \ket{0}_{D_v} \otimes \ket{1}_{D_h} + \ket{1}_{D_v} \otimes  \ket{0}_{D_h} \bigr) \ ,
\end{align*}
where we used the photon number  bases for $D_v$ and $D_h$. Provided the detectors are perfect, their action is defined for $D=D_{v,h}$ by the projectors  $\Pi_D^0 = \proj{0}_D$ and ${\Pi_D^1 = \proj{1}_D}$.  Since each of the two modes ($D_v$ and $D_h$)  is measured separately, the overall measurement is given by
\begin{align*}
\{\Pi_{D_v}^0 \otimes \Pi_{D_h}^0, \Pi_{D_v}^0 \otimes \Pi_{D_h}^1, \Pi_{D_v}^1 \otimes \Pi_{D_h}^0, \Pi_{D_v}^1 \otimes \Pi_{D_h}^1\} \ .
\end{align*}
Following~\eqref{eq:Born}, the  raw randomness $X = (X_v, X_h)$ is equivalent to a uniformly distributed bit, i.e., 
\begin{align*}
  P_X(0, 1) = P_X(1, 0) = {\textstyle \frac{1}{2}} \ .
\end{align*}
\end{example}

\begin{example}[Inefficient detector] \label{ex:inefficient}
A realistic photon detector detects an incoming photon only with bounded probability $\mu$. On the subspace of the optical mode,  $D=D_h$ or $D=D_v$, spanned by $\ket{0}_{D}$ (no photon) and $\ket{1}_{D}$ ($1$ photon), its action is given by the POVM $\{M^0_{D}, M^1_{D}\}$ with 
\begin{align*}
  M^1_{D} = \mu \proj{1}_{D} \quad \text{and} \quad M^0_{D} = \id_{D} - M^1_{D} \ . 
\end{align*}
To describe this as a projective measurement, we need to consider an extended space with an additional subsystem, $D'$, that determines whether the detector is sensitive or not (states $\ket{1}_{D'}$ and $\ket{0}_{D'}$, respectively). Specifically, we may define the extended projective measurement $\{\Pi_{D D'}^0, \Pi_{D D'}^1\}$ by
 \begin{align*}
  \Pi^1_{D D'} = \proj{1}_{D} \otimes \proj{1}_{D'}  \quad \text{and} \quad  \Pi_{D D'}^0 = \id_{D D'} - \Pi_{D D'}^1 
\end{align*}
It is then easily verified that the action of $\{M^0_{\mathrm{D}}, M^1_{\mathrm{D}}\}$ on any state $\sigma_{D}$ is reproduced by the action of $\{\Pi_{D D'}^0, \Pi_{D D'}^1\}$ on the product $\sigma_{D} \otimes \sigma_{D'}$ where
\begin{align} \label{eq:detectionsystem}
  \sigma_{D'} = (1-\mu) \proj{0}_{D'} + \mu \proj{1}_{D'} \ .
\end{align}
\end{example}

To assess the quality of a QRNG, we also need a description of its side information. While our modelling of QRNGs does not specify side information explicitly, the idea is to take into account any possible side information compatible with the model. To do so, we consider a purification $|\psi\rangle_{SE}$ of the state $\rho_S$ with purifying system $E$, as illustrated in Fig.~\ref{fig:model1}. Any possible side information $C$ may now be described as the outcome of a measurement on $E$.

\begin{figure}[h]
\begin{tikzpicture}[scale=0.9]
\scriptsize
\draw (0,0)--(3,0)--(3,2)--(0,2)--(0,0);
\draw (0.5,0.5)--(2.5,0.5)--(2.5,1.5)--(0.5,1.5)--(0.5,0.5);
\node at (1.5,1.25){Projective};
\node at (1.5,0.75){Measurement};
\draw[->] (2.5,3)--(1.5,2.05);
\draw[->] (2.5,3)--(3.5,2.05);
\draw[->] (1.5,0.25)--(1.5,-0.5);
\draw[->] (1.5,1.95)--(1.5,1.55);
\node at (2,-0.5){$X$};
\node at (2.5,3.3){$|\psi\rangle_{SE}$};
\node at (1.8,2.75){$S$};
\node at (3.2,2.75){$E$};
\end{tikzpicture}
\caption{\emph{Side information.} For a QRNG, defined by a projective measurement on a system $S$ with outcome $X$, all side information can be obtained from a purifying system $E$, i.e., an extra system that is chosen such that the joint state on $S$ and $E$ is pure.}
\label{fig:model1}
\end{figure}
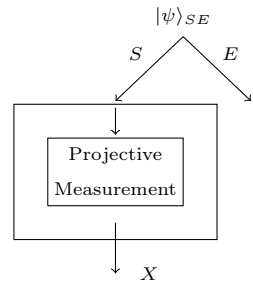

\begin{example}[Side information about inefficient detector] \label{ex:detectionbit}
Consider an inefficient detector $D=D_v$ or ${D=D_h}$ as in Example~\ref{ex:inefficient}. A classical bit $R$ may determine whether the detector is sensitive to incoming photons ($R = 1$) or not ($R = 0$). $R$ could then be considered as side information $W$. This information can indeed be easily obtained from a measurement on a purification of the state $\sigma_{D'}$ defined by~\eqref{eq:detectionsystem}. For example, for the purification
\begin{align*}
  \ket{\phi}_{D'  E} = \sqrt{1-\mu} \ket{0}_{D'} \otimes \ket{0}_E + \sqrt{\mu} \ket{1}_{D'} \otimes \ket{1}_E \ ,
\end{align*}
the value $R$ is retrieved as the outcome of the projective measurement $\{\proj{0}_E, \proj{1}_E\}$ applied to $E$.
\end{example}

The next statement is essentially a recasting of known facts about randomness extraction in the presence of quantum side information, combined with the fact that quantum theory is complete~\cite{colbeck11}. However, because it is central for our analysis, we formulate it as a proposition.

\begin{proposition} \label{pr:bound}
  Consider a QRNG that generates raw randomness $X$ and let $E$ be a purifying system of $S$. Furthermore, let $f$ be a function chosen uniformly at random and independently of all other values from a two-universal family of hash functions with output length \begin{align*}
  \ell \leq H_{\min}(X|E) - 2 \log(1/\epsilon) \ .
\end{align*}
Then the result $Z = f(X)$ is $\epsilon$-truly random. 
\end{proposition}

\begin{proof}[Proof Sketch]
According to the definition of $\epsilon$-true randomness we need to ensure that 
\begin{align} \label{eq:proofdistance}
\frac{1}{2} \bigl\| P_{F(X) W F}-P_{\bar{Z}} \times P_{W F}  \bigr\|_1 \leq \epsilon \ ,
\end{align}
where $W$ is any value that is available outside the future light cone of the event where the measurement $X$ is started,\footnote{Technically, we demand that $W$ is defined within some model compatible with quantum theory and with free choice (see~\cite{colbeck11}).} where $F$ is the random variable indicating the uniform choice of the hash function from the two-universal family, and where $P_{\bar{Z}}$ is the uniform distribution on $\{0, 1\}^{\ell}$. It follows from the completeness of quantum theory~\cite{colbeck11} that such a $W$ can always be obtained by a measurement of all available quantum systems, in our case $E$.\footnote{To see this consider $X$ as the random variable obtained by measuring $\{\Pi_S^x\otimes \id_E\}_{x\in\mathcal{X}}$ on the pure system $SE$. From the completeness it follows that $X$ cannot be predicted better within any extended theory than within quantum mechanics. Within quantum mechanics maximal information about $X$ is obtained from a measurement on the purifying system $E$.}  But because the trace distance can only decrease under physical mappings, \eqref{eq:proofdistance} holds whenever 
\begin{align*}
\frac{1}{2} \bigl\| \rho_{F(X) E F}-\rho_{\bar{Z}} \otimes \rho_{E F}  \bigr\|_1 \leq \epsilon \ .
\end{align*}
The claim then follows by the Leftover Hash Lemma (see Section~\ref{sec:hash}.)
\end{proof}

Note that Proposition~\ref{pr:bound} does not require a description of classical side information, i.e., there is no need to model the side information explicitly. This is important for practice, as it could be hard to find an explicit and complete model for all classical side information present in a realistic device. 

\subsection{Maximum classical noise model} 

Proposition~\ref{pr:bound} provides us with a criterion to determine the amount of true randomness that can be extracted from the output of a noisy QRNG. However, the criterion involves the conditional min-entropy for quantum systems, which may be hard to evaluate for practical devices. In the following, we are seeking for an alternative criterion that involves only classical quantities. The rough idea is to find a classical value $C$ which is as good as the side information $E$, in the sense that
\begin{align} \label{eq:conv}
H_{\min}(X|C)\leq H_{\min}(X|E)
\end{align}
holds. 

The random variable $C$ may be obtained by a measurement on the system $S$, but this measurement must not  interfere with the measurement carried out by the QRNG.  Furthermore, \eqref{eq:conv} can only hold if the measurement of $C$ is maximally informative. Technically, this means that the post-measurement state should be pure conditioned on $C$. This motivates the following definition (see Fig.~\ref{fig:model2}).

\begin{definition}\label{def:C}
A \emph{maximum classical noise model} for a QRNG with state $\rho_S$ and projective measurement $\{\Pi_S^x\}_x$ on $S$ is a generalised measurement\footnote{A \emph{generalised measurement} on $S$ is defined by a family of operators $\{E_S^c\}_c$ such that $\sum_c (E_S^c)^{\dagger} E_S^c = \id_S$.} $\{E_S^c\}_{c\in\mathcal{C}}$ on $S$ such that the following requirements are satisfied:
\begin{enumerate}
\item the map 
\begin{align*}
\mathcal{P}_{X\leftarrow S}: \quad \sigma_S\mapsto \sum_{x}\tr\left(\Pi_S^x\sigma_S\right)|x\rangle\langle x| \ ,
\end{align*}
is invariant under composition with the map
\begin{align*}
\mathcal{E}_{S\leftarrow S}:  \quad \sigma_S\mapsto \sum_c E_S^c\sigma_S (E_S^c)^{\dagger}
\end{align*}
 i.e., $\mathcal{P}_{X\leftarrow S}\circ\mathcal{E}_{S\leftarrow S}=\mathcal{P}_{X\leftarrow S}$;
\item  the state 
\[\rho_{S|C=c}=\frac{(E_S^c)^{\dagger}\rho_S E_S^c}{\tr((E_S^c)^{\dagger} \rho_SE_S^c)}\]
 obtained by conditioning on the outcome $C=c$  of the  measurement $\{E_S^c\}_{c}$ is pure, for any $c \in \mathcal{C}$.
\end{enumerate}
The outcome $C$ of the measurement $\{E_S^c\}_{c}$ applied to $\rho_S$ is called \emph{maximum classical noise}.  
\end{definition}

\begin{figure}[h]
\begin{tikzpicture}[scale=0.9]
\scriptsize
\draw (0,0)--(3,0)--(3,4)--(0,4)--(0,0);
\draw (0.5,0.5)--(2.5,0.5)--(2.5,1.5)--(0.5,1.5)--(0.5,0.5);
\draw (0.5,2.5)--(2.5,2.5)--(2.5,3.5)--(0.5,3.5)--(0.5,2.5);
\node at (1.5,1){\{$\Pi_S^x\}_x$};
% \node at (1.5,0.75){Measurement};
\node at (1.5,3){$\{E_S^c\}_{c}$};
\draw[->] (1.5,3.9)--(1.5,3.6);
\draw[->] (2.5,5)--(1.5,4.05);
\draw[->] (2.5,5)--(3.5,4.05);
\draw[->] (1.5,0.25)--(1.5,-0.5);
\draw[->,decorate,decoration={snake,post length=1 mm,amplitude=1mm,segment length=6mm}] (2.5,3)--(4,3);
\node at (4.25,3) {$C$};
\node at (2,-0.5){$X$};
\node at (2.5,5.3){$|\psi\rangle_{SE}$};
\node at (1.8,4.75){$S$};
\node at (3.2,4.75){$E$};
\draw[->] (1.5,2.4)--(1.5,1.7);
\node at (2.25,2) {$S'=S$};
\end{tikzpicture}
\caption{\emph{Classical noise model.} The maximum classical noise $C$ of a QRNG is defined by a measurement on $S$ that does not affect the projective measurement carried out by the QRNG, but gives maximal information about the raw randomness, $X$.}
\label{fig:model2}
\end{figure}
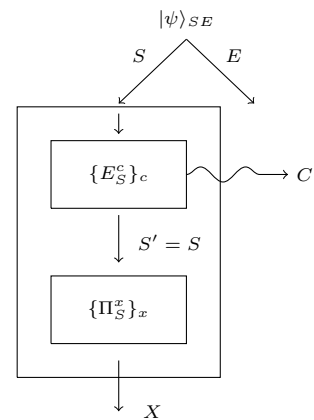

\begin{example}[Maximum classical noise model for inefficient detector] \label{ex:maxnoise}
  Consider again an inefficient detector as defined in Example~\ref{ex:inefficient} and its description in terms of a projective measurement $\{\Pi_{D D'}^0, \Pi_{D D'}^1\}$ on an extended system. If the state $\rho_{D}$ of the optical mode is pure then the measurement $\{E_{D D'}^{0}, E_{D D'}^{1} \}$ defined by 
\begin{align*}
   E_{D D'}^{0} = \id_D \otimes \proj{0}_{D'} \quad \text{and} \quad E_{D D'}^{1} = \id_D \otimes \proj{1}_{D'} 
\end{align*} 
is a maximum classical noise model. 
 To see this, note that the first criterion of Definition~\ref{def:C} is satisfied because this measurement commutes with the measurement $\{\Pi_{D D'}^0, \Pi_{D D'}^1\}$ of the detector. Furthermore, because $\{E_{D D'}^{0}, E_{D D'}^{1} \}$ restricted to $D'$ is a rank-one measurement, the post-measurement state is pure, so that the second criterion of Definition~\ref{def:C}  is also satisfied. Note that the maximum classical noise, $C$, defined as the outcome of the measurement $\{E_{D D'}^{0}, E_{D D'}^{1} \}$, is a bit that indicates whether the detector is sensitive or not, as in Example~\ref{ex:detectionbit}. For the PBS-based QRNG with two detectors, $D_v$ and $D_h$, the classical noise would be $C = (R_v, R_h)$, where $R_v$ and $R_h$ are the corresponding indicator bits for each detector. 
\end{example}

We remark that the definition of a maximum classical noise model is not unique. But this is irrelevant, for its main use is to provide a lower bound on $H_{\min}(X|E)$ and therefore (by virtue of Proposition~\ref{pr:bound}) on the number of truly random bits that can be obtained by hashing. 

\begin{lemma} \label{lem:mbound}
Consider a QRNG that generates raw randomness $X$ and let $E$ be a purifying system. Then, for any maximum classical noise $C$, 
\begin{align}\label{eq:mbound} 
  H_{\min}(X|C) \leq H_{\min}(X|E)  \ .
\end{align} 
\end{lemma} 

\begin{proof} The first requirement of Definition~\ref{def:C} guarantees that the random variables $C$ and $X$ are defined simultaneously. Because, by the second requirement of Definition~\ref{def:C}, the state of $S$ conditioned on $C$ is pure, it is necessarily independent of $E$. Since $X$ is obtained by a measurement on $S$, it is also independent of $E$, conditioned on $C$. Hence we have the Markov chain 
\begin{align*} 
X\leftrightarrow C\leftrightarrow E \ ,
\end{align*}
which implies 
\begin{align*}
  H_{\min}(X|C)= H_{\min}(X|CE) \ .
\end{align*}
The assertion then follows from the data processing inequality for the min-entropy~\eqref{eq:DPI}, 
\begin{align*}
  H_{\min}(X|C E) \leq H_{\min}(X|E) \ .
\end{align*}
\end{proof}

From the joint probability distribution determined by the Born rule
\begin{equation}\label{eq:Born2}
P_{XC}(x,c)=\tr(\Pi_S^xE_S^c\rho_S(E_S^c)^{\dagger}) \ ,
\end{equation}
the conditional min-entropy $H_{\min}(X|C)$ can be calculated using~\eqref{eq:condHmin1} and~\eqref{eq:Hmin}.

\begin{example}[Extractable randomness for a detector with efficiency $\mu$]
Given the maximum classical noise model from Example~\ref{ex:maxnoise}, we can easily calculate the conditional min-entropy of the measurement outcome $X$. For example, assuming that the optical mode $D$ carries with equal probability no or one photon, we find
\begin{align*}
H_{\min}(X|C)&=-\log[P_C(0)\cdot 1+P_C(1)\cdot\frac{1}{2}]\\
&=-\log[(1-\mu)\cdot 1+\mu\cdot\frac{1}{2}] \ .
\end{align*}
\end{example}

\subsection{Quantum randomness} 

When analysing realistic QRNGs, it is convenient to describe them in terms of purely classical random variables. As we have already seen above, side information can be captured most generally by a maximum classical noise model and, hence, a random variable $C$ (see Definition~\ref{def:C}). Similarly, we may introduce a random variable, $Q$, that accounts for the ``quantum randomness'', i.e., the part of the randomness that is intrinsically unpredictable. The idea is to define this as the randomness that ``remains'' after accounting for the maximum classical noise~$C$ (see Fig.~\ref{fig:model3}). 

\begin{definition} \label{def:QRM}
  Consider a QRNG that generates raw randomness $X$ and let $C$ be maximum classical noise, jointly distributed according to $P_{X C}$. Let $P_Q$ be a probability distribution and let $\chi: \, (q,c) \mapsto x$ be a function  such that
  \begin{align*}
    P_{X C} = P_{\chi(Q,C) C} \ ,
  \end{align*}
  where the distribution on the r.h.s.\ is defined by
  \begin{align*}
    P_{\chi(Q,C) C}(x,c) = \sum_{q: \, \chi(q,c) = x} P_{Q}(q) P_C(c) \ . 
  \end{align*}
  The corresponding random variable $Q$ is called \emph{quantum randomness}. 
\end{definition}

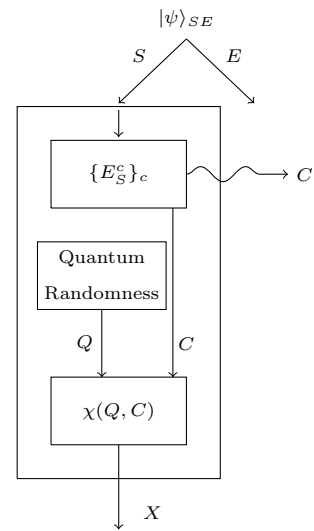
\begin{figure}[h]
\begin{tikzpicture}[scale=0.9]
\scriptsize
\draw (0,0.5)--(3,0.5)--(3,6)--(0,6)--(0,0.5);
\draw (0.5,4.5)--(2.5,4.5)--(2.5,5.5)--(0.5,5.5)--(0.5,4.5);
\draw[->] (1.5,5.95)--(1.5,5.55);
\node at (1.5,5){$\{E_S^c\}_{c}$};
\draw[->,decorate,decoration={snake,post length=1 mm,amplitude=1mm,segment length=6mm}] (2.5,5)--(4,5);
\draw[->] (2.3,4.5)--(2.3,2);
\node at (4.25,5){$C$};
\draw (0.3,3)--(2.2,3)--(2.2,4)--(0.3,4)--(0.3,3);
\node at (1.25,3.75){Quantum};
\node at (1.25,3.25){Randomness};
\node at (1,2.5){$Q$};
\node at (2.5,2.5){$C$};

\draw[->] (1.25,3)--(1.25,2);

\draw (0.5,1)--(2.5,1)--(2.5,2)--(0.5,2)--(0.5,1);

\draw[->] (2.5,7)--(1.5,6.05);
\draw[->] (2.5,7)--(3.5,6.05);

\node at (1.5,1.5){$\chi(Q,C)$};

\draw[->] (1.5,1)--(1.5,-0.25);

\node at (2,0){$X$};
\node at (2.5,7.3){$|\psi\rangle_{SE}$};
\node at (1.8,6.75){$S$};
\node at (3.2,6.75){$E$};

\end{tikzpicture}
\caption{\emph{Quantum randomness.} The raw randomness $X$ can be seen as a function $\chi$ of the quantum randomness $Q$ and  the classical noise $C$. This allows us to replace the real device from Fig.~\ref{fig:model1} by a model based on classical random variables.}
\label{fig:model3}
\end{figure}

\begin{example}[Quantum randomness of the PBS-based QRNG]\label{ex:qrand}

  The quantum randomness of the PBS-based QRNG of Example~\ref{ex:idealPBS} may be defined as the path that the photon takes after the PBS  (i.e., whether it travels to $D_v$ or $D_h$).  For a single diagonally polarised photon, $Q$ would therefore be a uniformly distributed bit. Then, for  inefficient detectors with maximum classical noise $R_v$ and $R_h$ defined as in  Example~\ref{ex:maxnoise}, the function $\chi: \, (q, r_v, r_h) \mapsto x=(x_v, x_h)$ is given by 
\begin{align*}
  \chi(q, r_v, r_h) = \begin{cases} (r_v, 0) & \text{if $q = v$} \\ (0, r_h) & \text{if $q=h$.} \end{cases}
  % \chi(0,0,0) & = (0,0) \\
  % \chi(0,0,1) & = (0,0) \\
  % \chi(0,1,0) & = (1,0) \\
  % \chi(0,1,1) & = (1,0) \\
  % \chi(1,0,0) & = (0,0) \\
  % \chi(1,0,1) & = (0,1) \\
  % \chi(1,1,0) & = (0,0) \\
  % \chi(1,1,1) & = (0,1)  \ .
\end{align*}
\end{example}

\section{Example: analysis of a noisy PBS-based QRNG}\label{example}

To illustrate the effect of noise on true randomness generation, we study, as an example, a PBS-based QRNG with two detectors. A realistic description of this QRNG would take into account that the photon detectors are subject to dark counts and cross talk, and that their efficiency generally depends on the number of incoming photons. While an analysis based on such a more realistic model is provided in Appendix~\ref{quantis}, we consider here a simplified model where the two detectors, $D_v$ and $D_h$, are assumed not to click if there is no incoming photon and click with constant probability $\mu$ in the presence of one or more incoming photons. Fig.~\ref{QRNG} schematically illustrates the working of our QRNG. The raw randomness consists of bit pairs $X=(X_v,X_h)\in\{00,01,10,11\}$, where

\begin{equation}\label{eq:PBS}
X_{v,h}=\left\{\begin{array}{ll}1&\text{if } D_{v,h} \text{ clicks.}\\0&\text{else}\end{array}\right.
\end{equation}

\begin{figure}[h]
\begin{tikzpicture}
\scriptsize
\draw[->,decorate,decoration={snake,post length=1 mm,amplitude=1mm,segment length=6mm}] (1,0)--(3.8,0);
\draw[line width=0.4mm] (3.75,-0.5) --(4.25,0.5);
\draw[->,decorate,decoration={snake,post length=1 mm,amplitude=1mm,segment length=6mm}] (4,0.2)--(4,2);
\draw[->,decorate,decoration={snake,post length=1 mm,amplitude=1mm,segment length=6mm}] (4.2,0)--(6,0);
\draw (4.5,3) arc(0:180:0.5);
\draw (4.5,3)--(4.5,2.2)--(3.5,2.2)--(3.5,3);
\draw (7,-0.5) arc(-90:90:0.5);
\draw (7,-0.5)--(6.2,-0.5)--(6.2,0.5)--(7,0.5);
\node at (2.2,0.25) {$n$};
\node at (0.5,1.8) {\normalsize{$\boxed{P_N(n)=e^{-|\alpha|^2}\frac{|\alpha|^{2n}}{n!}}$}};
\node at (5.2,0.25) {$m$};
\node at (3.4,1.1) {$n-m$};
\node at (4,2.8) {{\normalsize $D_v$}};
\node at (6.8,0)  {{\normalsize $D_h$}};
\end{tikzpicture} 
\caption{\emph{Noisy PBS-based QRNG.}  A source emits pulses of photons. The photon numbers are typically following a Poisson distribution, $P_N$.}\label{QRNG}
\end{figure}
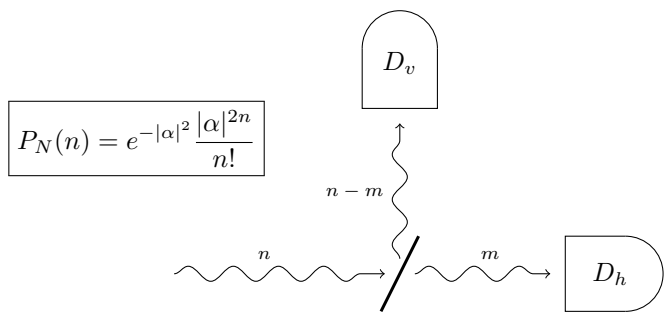
For our model, we assume that the source emits $n\in\{0,\dots,\infty\}$ photons according to the Poisson distribution\footnote{Note that, for realistic sources (such as lasers or LEDs), the photon numbers between subsequent pulses may be correlated, e.g., due to photon (anti-)bunching. However, because we model the source as a mixed density operator over Fock states, we automatically include knowledge about the exact photon number as side information, i.e., the extracted randomness will be uniform even if this number is known. Furthermore, the effect of bunching on the overall frequency of the photon numbers (which we assume here to follow a Poisson distribution) is usually much smaller than other imperfections, so that it can be safely ignored in our analysis.}
\begin{align}
P_N(n)=e^{-|\alpha|^2}\frac{|\alpha|^{2n}}{n!}\; .
\label{eq:poisson}
\end{align}
The outcome $x=(1,1)$ occurs if the following conditions are satisfied:
\begin{itemize}
\item[-] The source emits at least two photons.
\item[-] After the interaction with the PBS photons are in both paths. This happens with probability 
\[1-2\left(\frac{1}{2}\right)^n,\]
because $\left(\frac{1}{2}\right)^n$ is the probability that all photons end up in the same path.
\item[-] Both detectors are sensitive (which happens with probability $\mu^2$.
\end{itemize}
Therefore the probability to obtain $x=(1,1)$ is given by
\[P_X(1,1)=\sum\limits_{n=2}^{\infty}P_N(n)\left(1-2\cdot\Big(\frac{1}{2}\Big)^n\right)\cdot\mu^2.\]
The remaining probabilities are determined by analogous considerations. They are summarised in Table~\ref{distrwithoutSI} in Appendix~\ref{HXRN}.\\
Following our framework, we start by modelling the device according to Definition~\ref{def:QRNG}, i.e., we define an input state $\rho_S$ and a projective measurement such that $P_X$ is reproduced by the Born rule \eqref{eq:Born}.

\subsection{Definition of the QRNG}

To define the QRNG, we need to specify a density operator $\rho_S$ (corresponding to the state before measurement) and a measurement $\{\Pi^x_S\}$. We start with the description of the density operator. 

Analogous to Example~\ref{ex:inefficient} we consider an extended space with additional subsystems that determine the number of incoming photons and whether the two detectors are sensitive.
\begin{itemize}
\item A subsystem $I$ encodes the intensity of the source, in terms of the photon number, $n$,  in states $|n\rangle$ with $n\in\{0,\ldots,\infty\}$, where $n$ is distributed according to the Poisson distribution \eqref{eq:poisson}.
\item For each of the $n$ photons emitted by the source the two light modes $D_v$ and $D_h$ travelling to the respective detectors have state $|\phi\rangle_{D_vD_h}$, as defined in Example~\ref{ex:idealPBS}.
\item Two subsystems $D'_v$ and $D'_h$, prepared in states $|r_{v}\rangle$ and $|r_{h}\rangle$ (see Example~\ref{ex:inefficient}) determine whether the respective detectors are sensitive ($r_{v,h}=1$) or not ($r_{v,h}=0$).
\end{itemize}
The state $\rho_S$ of the total system is thus given by
\begin{equation*}
\sum_{n,r_v,r_h}P_N(n)P_{R_v}(r_v)P_{R_h}(r_h)\proj{n} \otimes\proj{r_v} \otimes \proj{r_h} \otimes \proj{\phi} ^{\otimes n},
\end{equation*}
where $P_{R_{v,h}}(1)=\mu$ and where we omitted the subscripts for simplicity.

Note that by writing the state of the $n$ photons in tensor product form $|\phi\rangle^{\otimes n}_{D_vD_h}$, it is assumed that they are distinguishable particles, even though photons are fundamentally indistinguishable. However, it turns out that photons behave in beam-splitting experiments as if they were in-principle distinguishable (see for example~\cite{leonard03}).

To define the measurement $\{\Pi_{S}^{x}\}_{x\in\{11,10,01,00\}}$, let us consider each detector individually. Looking at $D_v$ the detector clicks if it is sensitive and if there is at least one photon in the corresponding path. For $n$ incoming photons, the latter criterion is determined by the two operators $\{P_v^{n,0},P_v^{n,1}\}$, where
\begin{align*}
P_v^{n,0}:=|0\rangle\langle 0|_{D_v}^{\otimes n} \quad \text{and} \quad 
P_v^{n,1}:=\id_{D_v}^{\otimes n}-|0\rangle\langle 0|_{D_v}^{\otimes n}
\end{align*}
correspond to the two cases where no, or at least one, photon is in the path going to $D_v$, respectively. For the other detector, $D_h$, we define $\{P_h^{n,0},P_h^{n,1}\}$ analogously. The measurement projectors are then given by
\begin{equation*}
\begin{array}{lll}
\Pi_{S}^{11}=\sum\limits_{n\in\{0,\infty\}}&&\hspace{-0.5cm}|n,1,1\rangle\langle n,1,1|_{ID'_vD'_h}\otimes P_v^{n,1}\otimes P_h^{n,1}\\[0.8cm]

\Pi_{S}^{10}=\sum\limits_{n\in\{0,\infty\}}&\Big[&|n,1,1\rangle\langle n,1,1|_{ID'_vD'_h}\otimes P_v^{n,1}\otimes P_h^{n,0}\\
&+&|n,1,0\rangle\langle n,1,0|_{ID'_vD'_h}\otimes P_v^{n,1}\otimes\id_{D_h}^{\otimes n}\Big]\\[0.8cm]

\Pi_{S}^{01}=\sum\limits_{n\in\{0,\infty\}}&\Big[&|n,1,1\rangle\langle n,1,1|_{ID'_vD'_h}\otimes P_v^{n,0}\otimes P_h^{n,1}\\
&+&|n,0,1\rangle\langle n,0,1|_{ID'_vD'_h}\otimes \id_{D_v}^{\otimes n}\otimes P_h^{n,1}\Big]\\[0.8cm]

\Pi_{S}^{00}=&&\hspace{-1.5cm}\id_S-\Pi_{S}^{11}-\Pi_{S}^{10}-\Pi_{S}^{01}.
\end{array}
\end{equation*}

The probability distribution $P_X$ of  the raw randomness, as obtained by applying the Born rule~\eqref{eq:Born}  to this state and measurement, is shown in Table~\ref{distrwithoutSI} of Appendix~\ref{HXRN}.

\subsection{Maximum classical noise model for the QRNG}

A possible maximum classical noise model for the QNRG is $\{E_S^{nr_vr_h} \}_{nr_vr_h}$ defined by\footnote{Note that the subsystems carrying the states of $D_h$ and $D_s$ should be interpreted as Fock spaces.}
\begin{align*}
E_{S}^{nr_vr_h} =   \proj{n}_{I} \otimes \proj{r_v}_{D'_v} \otimes \proj{r_h}_{D'_h}\otimes\id_{D_vD_h}^{\otimes n}
\end{align*} 

The classical noise, defined as the outcome of the measurement $\{E_S^{nr_vr_h} \}_{nr_vr_h}$, are the following three random variables
\begin{itemize}
\item $N$ with values  $n\in\{0,\ldots,\infty\}$ distributed according to the Poisson distribution \eqref{eq:poisson}. It corresponds to the side information about the number of photons emitted by the source.
\item $R_{v,h}$ with outcomes $r_{v,h}\in\{0,1\}$ distributed according to $P_{R_{v,h}}$. These two random variables encode the side information about the sensitivity of the two detectors.
\end{itemize}

As in Example~\ref{ex:maxnoise} this is a maximum classical noise model. The first criterion of Definition~\ref{def:C} is satisfied because the measurement $\{E_S^{nr_vr_h} \}_{nr_vr_h}$ commutes with the measurement  $\{\Pi_{S}^{x}\}_{x\in\{11,10,01,00\}}$. The second criterion of Definition~\ref{def:C}  is also satisfied, because $\{E_S^{nr_vr_h} \}_{nr_vr_h}$ restricted to $ID'_vD'_h$ is a rank-one measurement and because the state on the systems $D_vD_h$ is pure, so the post-measurement state is pure as well.

The total classical noise is the joint random variable $C=NR_vR_h$. The joint probability distribution $P_{XC}$ is given by the Born rule~\eqref{eq:Born2}, from which we obtain the conditional probability distribution $P_{X|C}$. It is summarised in Table~\ref{distrwithSI} of Appendix~\ref{HXRN}. This then allows us to calculate  $H_{\min}(X|C)$ using \eqref{eq:Hmin}, giving us a lower bound for $H_{\min}(X|E)$ according to Lemma~\ref{lem:mbound}. By virtue of Proposition~\ref{pr:bound} it is therefore a lower bound on the extractable true randomness. Fig.~\ref{fig:H} shows this bound for the specific value of $\mu=0.1$. An upper bound is given by the Shannon entropy $H(X|C)$ (see Appendix~\ref{vN}). The true value of the extractable randomness lies therefore somewhere in the blue shaded area. For comparison the unconditional min-entropy $H_{\min}(X)$ is shown, which corresponds to the extractable rate of uniformly distributed bits. The corresponding calculations can be found in Appendix~\ref{HXRN}.

\begin{figure}
\includegraphics[scale=0.6]{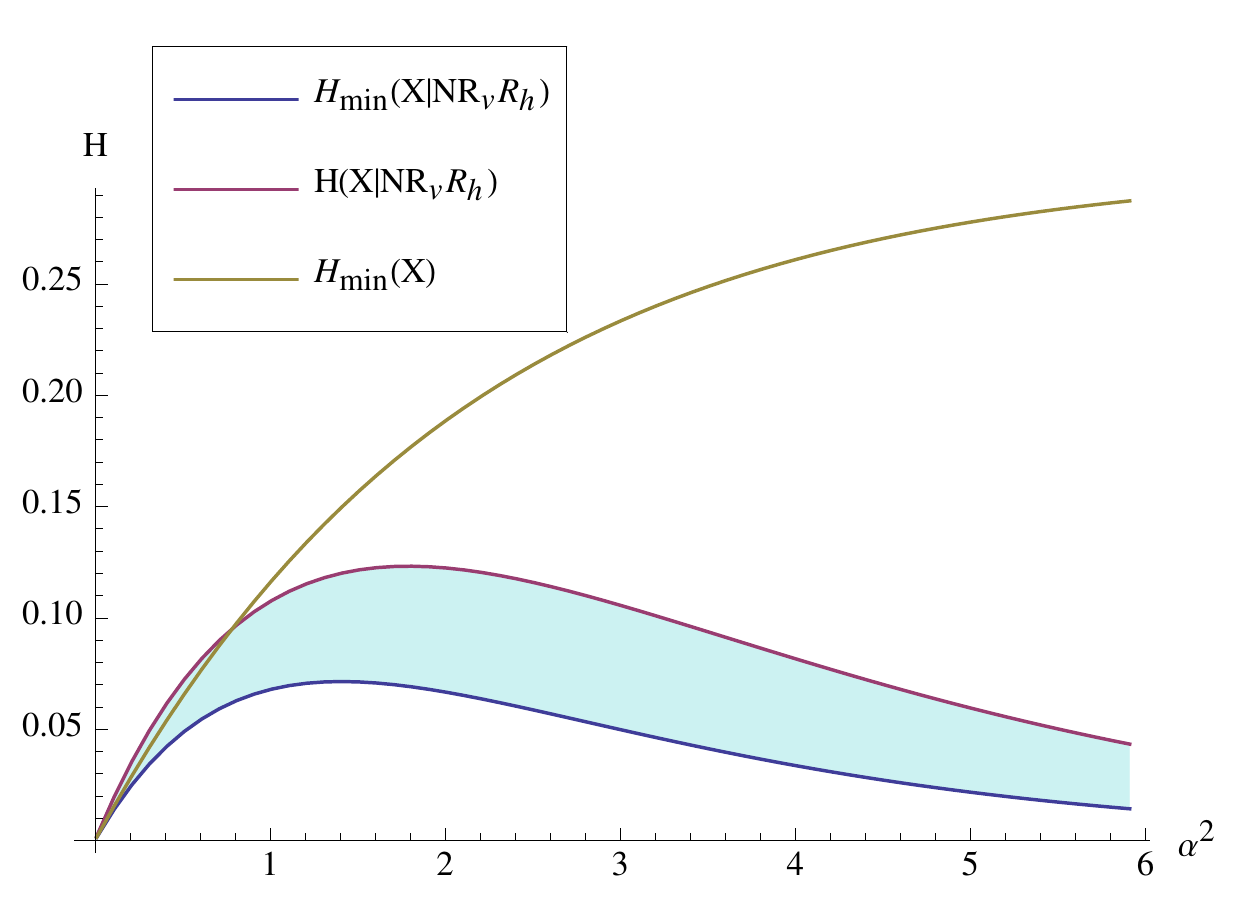}
\caption{\emph{Bounds for the extractable true randomness (for $\mu=0.1$).} The min-entropy, $H_{\min}(X|NR_vR_h)$, of the raw randomness corresponds to a lower bound for the extractable rate of truly random bits. The upper bound is given by the Shannon Entropy $H(X|NR_vR_h)$. Therefore, the amount of true randomness lies in the blue area. For comparison $H_{\min}(X)$ is shown, which corresponds to the extractable rate of uniformly distributed  (but not necessarily truly random) bits.}\label{fig:H}
\end{figure}

It can be seen that $H_{\min}(X)$ reaches a maximum value in the high intensity regime of  approximately $-\log_2((1-\mu)^2)$ which corresponds to the logarithm of the guessing probability for the most likely outcome $(0,0)$. Crucially, however, there is almost no true randomness left in this regime. This reflects the fact that an adversary having access to the information whether the detectors are sensitive or not can guess the outcome with high probability. In fact, in this regime the raw randomness will be almost independent of the quantum process, but only depend on the behaviour of the detectors. Therefore, the device does actually not correspond to a PBS-based QRNG but rather to an RNG based on (potentially classical) noise.

\subsection{Quantum randomness}

Analogously to Example~\ref{ex:qrand} the quantum randomness may be defined as the path the photons take after the PBS. The difference is that now there is not exactly one incoming photon but $n\in\{0,\ldots,\infty\}$. For each photon the quantum randomness, $Q$, is still a uniformly distributed bit $q\in\{h,v\}$. Because the number of incoming photons is not fixed, we define the quantum randomness by sequence of random variables $Q^{\infty} = (Q_1, Q_2, \ldots, )$, with distribution given by
\begin{align*}
P_{Q_i|Q_{i-1}}(q)=\frac{1}{2} \ .
\end{align*}
Together with the function 
\begin{align*}
\chi:(r_v,r_h,n,q_1,q_2,\ldots)\mapsto (x_h,x_v)  \ ,
\end{align*}
where
\begin{align*}
x_{v}=\begin{cases} 1&\text{if } r_{v}=1 \text{ and } |\{i:q_i=v\}_{i\leq 1\leq n}|\geq 1\\0&\text{else }\end{cases}
\end{align*}
and, likewise for $x_h$, this satisfies Definition~\ref{def:QRM}.

\section{Conclusions}

For randomness to be usable in applications, e.g., for drawing the numbers of a lottery, an important criterion is that it is unpredictable for everyone. Unpredictability is however not a feature of individual values or their frequency distribution, and can therefore not be certified by statistical tests. Rather, unpredictability is a property of the process that generates the randomness. This idea is captured by the notion of ``true randomness'' (see Definition~\ref{def_random}). The definition demands that the output of the process is independent of all side information available when the process is started.

While certain ideal quantum processes are truly random, practical Quantum Random Number Generators (QRNGs) are usually not. The reason is that, due to imperfections, the raw output of realistic devices depends on additional degrees of freedom, which can in principle be known beforehand.  In this work we showed how to model such side information and account for it in the post-processing of the raw randomness. We hope that our framework is useful for the design of next-generation QRNGs that are truly random. 

\begin{acknowledgments}
The authors thank Nicolas Gisin, Volkher Scholz, and Damien Stucki  for discussions and insight. We are also grateful for the collaboration with IDQ. This project was funded by the CREx project and supported by SNSF through the National Centre of Competence in Research ``Quantum Science and Technology'' and through grant No.\ 200020-135048, and by the European Research Council through grant No.~258932.
\end{acknowledgments}
\newpage
\begin{appendix}\label{appendix}
\onecolumngrid

\section{Block-wise Hashing and the Statistical Error of the Seed}\label{sec:seed}

In this section it is shown that the statistical error related to the choice of the hash function is not multiplied with the number of blocks in the case of block-wise post-processing. This implies that the hash function can be chosen once and therefore in principle be hard-coded in the device. 

Note that we describe this argument for the case of classical side information, $C$. However, replacing all probability distributions by density operators, the argument can be easily generalised to the case of quantum side information. 

Let
\[f:\mathcal{X}\rightarrow\{0,1\}^{\ell}\]
be the selected hash function. Assume that we apply the function to $k$ blocks, corresponding to a random variable $X_1\ldots X_k$, where each $X_i$ is a random variable with a alphabet $\mathcal{X}$. The final distribution $f(X_1)\ldots f(X_k)$ is the concatenation of the $k$ post-processed blocks.\\
In practice the hash function is not chosen according to a perfectly random distribution $U_F$ and there is also a possible correlation to the source. Therefore, we define the statistical error related to the seed by
\begin{equation}
\epsilon_{\text{seed}}:=\|P_{X_1\ldots X_kCF}-P_{X_1\ldots X_k C}\times U_F\|_1,
\end{equation}
We will show that the following bound holds
\begin{equation}
\|P_{f(X_1)\ldots f(X_k)CF}-U_\ell^k\times P_C\times U_F\|_1\leq k\epsilon_{\text{hash}} + \epsilon_{\text{seed}},
\label{bound}
\end{equation}
if $H_{\min}(X_i|X_{i-1}\ldots X_1C)\geq \ell$. (The lower bound on the entropy is automatically satisfied if $X_{i-1}\ldots X_1$ is considered as previously available side information.) Here $U_\ell$ is the uniform distribution on $\{0,1\}^{\ell}$.
\begin{proof}
From the Leftover Hash Lemma with Side Information \eqref{lhl} if follows that for $1\leq i \leq k$
\begin{equation}
\|P_{f(X_i)X_{i+1}\ldots X_k C}\times U_F-U_\ell\times P_{X_{i+1}\ldots X_kC}\times U_F\|_1\leq \epsilon_{\text{hash}}.
\end{equation}
From the fact that the trace distance can only decrease under the application of $f$ (see for example \cite{nielsen00}) we first observe that
\begin{equation}
 \|P_{f(X_1)\ldots f(X_k)CF}-P_{f(X_1)\ldots f(X_k) C}\times U_F\|_1\leq\epsilon_{\text{seed}}
\label{eq2}
\end{equation}
and
\begin{equation}
\|P_{f(X_i)f(X_{i+1})\ldots f (X_k )C}\times U_F-U_\ell\times P_{f(X_{i+1})\ldots f(X_k)C}\times U_F\|_1\leq\epsilon_{\text{hash}},
\label{eq3}
\end{equation}
holds.\\
Now we use the triangle inequality to bound the quantity we are interested in
\begin{align*}
&\|P_{f(X_1)\ldots f(X_k)CF}-U_\ell^k\times P_C\times U_F\|_1\\
\leq& \|P_{f(X_1)\ldots f(X_k)CF}-P_{f(X_1)\ldots f(X_k)C}\times U_F\|_1+\|P_{f(X_1)\ldots f(X_k)C}\times U_F-U_\ell^k\times P_C\times U_F\|_1.
\end{align*}
From Eq.~\eqref{eq2} it follows that the first term is smaller than $\epsilon_{\text{seed}}$.\\
The second term can be bounded by subsequent application of the triangle inequality
\begin{align*}
&\|P_{f(X_1)\ldots f(X_k)C}\times U_F-U_\ell^k \times P_C\times U_F\|_1\\
\leq & \|P_{f(X_1)\ldots f(X_k)C}\times U_F-U_\ell\times P_{f(X_2)\ldots f(X_k)C}\times U_F\|_1+\|P_{f(X_2)\ldots f(X_k)C}\times U_F-U_\ell^{k-1}\times P_C\times U_F\|\\
\leq &\sum_{i=1}^k \|P_{f(X_i)\ldots f(X_k)C}\times U_F-U_\ell\times P_{f(X_{i+1})\ldots f(X_k)C}\times U_F\|_1\\
\leq &k\epsilon_{\text{hash}},
\end{align*}
where we used Eq.~\eqref{eq3}. Combining the two bounds yields Eq.~\eqref{bound}.
\end{proof}

\section{On POVMs and their Decompositions into Projective Measurements}\label{POVM}

In Section~\ref{sec:QRNGmodel} we discussed how a QRNG can be modelled by an input state and a set of measurements on it. As explained there side information about the measurement outcome can either result from a mixed input state (vs.\ a pure one) or if the measurement is a POVM (vs.\ a projective measurement). In the framework presented in the following all the side information was associated to the input state, i.e., the measurement is assumed to be projective. Another approach would be to associate the side information with the measurement, by allowing it to be a POVM, and choosing a pure input state. The idea is then that a general POVM can be regarded as a mixture over projective measurements and that such a mixing is equivalent to a hidden variable model producing noise of classical nature. Therefore, a possible approach would be to start with a POVM and consider a specific decomposition. The adversary is then assumed to know which of the projective measurements was chosen. Such a decomposition is not unique, but one could hope that different decompositions yield the same side information. However, the following example shows that this is not true, i.e., the the amount of extractable randomness can be different for different decompositions.\\

Consider the POVM given by
\[\{M^0,M^1\}=\left\{\left(\begin{array}{cc}2/3&0\\0&1/3\end{array}\right)\right.,\left.\left(\begin{array}{cc}1/3&0\\0&2/3\end{array}\right)\right\}.\]
One possible decomposition is
\[\{\mathcal{P}^1,\mathcal{P}^2\}=\Bigg\{\left\{\left(\begin{array}{cc}1&0\\0&0\end{array}\right)\right.,\left.\left(\begin{array}{cc}0&0\\0&1\end{array}\right)\right\},\left\{\left(\begin{array}{cc}0&0\\0&1\end{array}\right)\right.,\left.\left(\begin{array}{cc}1&0\\0&0\end{array}\right)\right\}\Bigg\}.\]
such that for $x\in\{0,1\}$
\[\{M^0,M^1\}=\frac{2}{3}\mathcal{P}^1+\frac{1}{3}\mathcal{P}^2.\]
And another decomposition is
\[\{\tilde{\mathcal{P}}^1,\tilde{\mathcal{P}}^2,\tilde{\mathcal{P}}^3\}=\Bigg\{\left\{\left(\begin{array}{cc}1&0\\0&0\end{array}\right)\right.,\left.\left(\begin{array}{cc}0&0\\0&1\end{array}\right)\right\},\left\{\id\right.,\left.0\right\},\left\{0\right.,\left.\id\right\}\Bigg\},\]
with
\[\{M^0,M^1\}=\frac{1}{3}\tilde{\mathcal{P}}^1+\frac{1}{3}\tilde{\mathcal{P}}^2+\frac{1}{3}\tilde{\mathcal{P}}3.\]
Let now $Z$ be the random variable corresponding to the first composition i.e.\ $Z$ has outcomes $z\in\{1,2\}$ with $P_{Z}(z=1)=\frac{2}{3}$ and $P_{Z}(z=2)=\frac{1}{3}$ and $z=1$ means that $\mathcal{P}^1$ is applied. Analogously we define $\tilde{Z}$.\\

Consider the input state
\[|\Psi\rangle=\frac{1}{\sqrt{2}}\left(|0\rangle+| 1\rangle\right).\]
A straight forward calculation shows that
\[2^{-H_{\min}(X|Z)}=\frac{1}{2}\]
whereas
\[2^{-H_{\min}(X|\tilde{Z})}=\frac{5}{6}.\]
In other words this means that the second decomposition gives more side information to a potential adversary and therefore, corresponds to less extractable randomness.

\section{Upper Bound for the Extractable Randomness}\label{vN}
In this section we show that an upper bound for the extractable entropy which is independent of $C$ is given by the Shannon entropy 
\[H(X|C)=\sum_{c\in\mathcal{C}}P_C(c)H(X|C=c).\]
More precisely we show that for any function ${f:\mathcal{X}\rightarrow\{0,1\}^{\ell}}$ such that $\|P_{f(X)C}-U_\ell\times P_{C}\|_1\leq\epsilon$, it holds that
\begin{equation}
\ell\leq H(X|C)+4\epsilon\log \ell+2h(\epsilon),
\end{equation}
where $h(x)=-x\log x-(1-x)\log(1-x)$ is the binary entropy function. The Shannon entropy of a probability distribution $P_X$ is defined as
\[H(X)=\sum_x -P_X(x)\log P_X(x)\]
and the conditional Shannon entropy is given by
\[H(X|C)=H(XC)-H(C).\]
\begin{proof}
We first show that the Shannon entropy of $X$ can only decrease under the application of $f$
\[H(X)\geq H(f(X)).\]
This can be seen from the observation that $H(f(X)|X)=0$ which implies $H(X)=H(f(X)X)$. It then follows
\[H(X|f(X))+H(f(X))=H(X).\]
Using $H(X|f(X))\geq 0$ gives the inequality $H(X)\geq H(f(X))$, which generalises to
\begin{equation}
H(X|C)\geq H(f(X)|C).
\label{eq1}
\end{equation}
Now the continuity of the conditional entropy \cite{alicki04} can be used, yielding
\[|H(f(X)|C)-H(U_\ell|C)|\leq 4\epsilon\log {\ell}+2h(\epsilon)\]
and therefore
\[H(f(X)|C)\geq {\ell}-4\epsilon\log {\ell}-2h(\epsilon).\]
Combination with Eq.~\eqref{eq1} yields the desired statement.
\end{proof}

\section{Calculation of $H_{\min}(X|NR_vR_h)$ and $H_{\min}(X)$ for the PBS-based QRNG}\label{HXRN}

This section provides additional details for the example discussed in Section~\ref{example}. 

The conditional min-entropy $H_{\min}(X|NR_vR_h)$ is given by
\begin{equation}
2^{-H_{\min}(X|NR_vR_h)}=\sum\limits_{nr_vr_h}P_N(n)P_{R_v}(r_v)P_{R_h}(r_h)2^{-H_{\min}(X|nr_vr_h)},
\label{eq:condHmin}
\end{equation}
where
\begin{equation}
H_{\min}(X|nr_vr_h)=-\log_2\left[\max\limits_{x=(x_v,x_h)}P_{X|NR_vR_h}(x|nr_vr_h)\right].
\label{Hmin}
\end{equation}
Table~\ref{distrwithSI} summarises the guessing probabilities $p(x_vx_h|nr_vr_h)$ for different $n$ and $r_{v,h}$.
\begin{center}
\begin{equation*}
H_{\min}(X|NR_vR_h)=-\log\left[P_N(0)+\sum_{n=1}^{\infty}P_N(n)\left\{(1-\mu)^2+2\cdot\mu(1-\mu)\left(1-\Big(\frac{1}{2}\Big)^n\right) +\mu^2\max\left[\Big(\frac{1}{2}\Big)^n,1-2\cdot\Big(\frac{1}{2}\Big)^n\right]\right\} \right]
\end{equation*}
\begin{table}[h]
{\renewcommand{\arraystretch}{3}
\renewcommand{\tabcolsep}{0.6cm}
\begin{tabular}{|c|c|}
\hline
$x_v,x_h$&$P_X(x_vx_h)$\\
\hline
$(1,1)$&$\sum\limits_{n=2}^{\infty}P_N(n)\left(1-2\cdot\Big(\frac{1}{2}\Big)^n\right)\cdot\mu^2$\\
\hline
$(0,1)$&$\frac{1}{2}P_N(1)\cdot\mu+\sum\limits_{n=2}^{\infty}P_N(n)\left(\Big(\frac{1}{2}\Big)^n\cdot\mu+\left(1-\Big(\frac{1}{2}\Big)^n\right)\cdot\mu(1-\mu)\right)$\\
\hline
$(1,0)$&$\frac{1}{2}P_N(1)\cdot\mu+\sum\limits_{n=2}^{\infty}P_N(n)\left(\Big(\frac{1}{2}\Big)^n\cdot\mu+\left(1-\Big(\frac{1}{2}\Big)^n\right)\cdot\mu(1-\mu)\right)$\\
\hline
$(0,0)$&$P_N(0)+P_N(1)\cdot(1-\mu)+\sum\limits_{n=2}^{\infty}P_N(n)\left(2\cdot\Big(\frac{1}{2}\Big)^n\cdot(1-\mu)+\left(1-2\cdot\Big(\frac{1}{2}\Big)^n\right)\cdot(1-\mu)^2\right)$\\[0.1cm]
\hline

\end{tabular}
}
\caption{\emph{Statistics of the raw randomness.} Distribution of the QRNG output $X$ without conditioning on side information.}
\label{distrwithoutSI}
\end{table}
\begin{table}[h]
{\renewcommand{\arraystretch}{1.5}
\renewcommand{\tabcolsep}{0.3cm}
\begin{tabular}{|c|c|c|c|c|c|}
\hline
$r_v,r_h$&$n$&$P_{X|NR_vR_h}(00|nr_vr_h)$&$P_{X|NR_vR_h}(01|nr_vr_h)$&$P_{X|NR_vR_h}(10|nr_vr_h)$&$P_{X|NR_vR_h}(11|nr_vr_h)$\\[0.1cm]
\hline
$(\cdot,\cdot)$&$0$&$1$&$0$&$0$&$0$\\
\hline
$(0,0)$&$\geq 1$&$1$&$0$&$0$&$0$\\
\hline
(0,1)&$\geq 1$&$\Big(\frac{1}{2}\Big)^n$&$1-\Big(\frac{1}{2}\Big)^n$&$0$&$0$\\
\hline
(1,0)&$\geq 1$&$\Big(\frac{1}{2}\Big)^n$&$0$&$1-\Big(\frac{1}{2}\Big)^n$&$0$\\
\hline
(1,1)&$\geq 1$&$0$&$\Big(\frac{1}{2}\Big)^n$&$\Big(\frac{1}{2}\Big)^n$&$1-2\cdot\Big(\frac{1}{2}\Big)^n$\\
\hline
\end{tabular}
}
\caption{\emph{Raw randomness conditioned on side information.} Probability distribution of the QRNG output $X$ conditioned on the side information $R_{v,h}$ and $N$.}
\label{distrwithSI}
\end{table}

\end{center}

The Shannon entropy $H(X|RN)$, which corresponds to an upper bound for the extractable entropy as shown in Section~\ref{vN} of the Appendix. It is equal to
\begin{align*}
H(X|NR_vR_h)=&\sum_{n,r_v,r_h}P_N(n)P_{R_v}(r_v)P_{R_h}(r_h)H(X|nr_vr_h)\\
=&\sum_{n=1}^{\infty}P_N(n)\Big\{2\cdot \mu\cdot(1-\mu)\left[-\Big(\frac{1}{2}\Big)^n\log\left(\Big(\frac{1}{2}\Big)^n\right)-\left(1-\Big(\frac{1}{2}\Big)^n\right)\log\left(\left(1-\Big(\frac{1}{2}\Big)^n\right)\right)\right]\\
&\mu^2\left[-2\cdot\Big(\frac{1}{2}\Big)^n\log\left(\Big(\frac{1}{2}\Big)^n\right)-\left(1-2\cdot\Big(\frac{1}{2}\Big)^n\right)\log\left(\left(1-2\cdot\Big(\frac{1}{2}\Big)^n\right)\right)\right]\Big\}
\end{align*}

\section{A More Detailed Model for the PBS-based QRNG}\label{quantis}

In this section we consider a more detailed model for the beam-splitter based QRNG considered in Section~\ref{example}. Now the sensitivity of the detectors with efficiency $\mu$ is assumed to depend on the number of photons hitting it (in Section~\ref{example} we assumed that it is independent). Explicitly we assume that if the source emits $n$ photons and $0\leq m\leq n$ photons arrive at one of the detectors $D_{v,h}$ the probability that the detector not fire is equal to $(1-\mu)^m$. In this more realistic model we also take noise in form of dark counts, afterpulses and crosstalk into account. As in Example~\ref{example} one can define a state and measurements such that the side information corresponding to the noise is encoded in a maximum classical noise model. Because those definitions are straightforward and add nothing conceptually new to the example, we proceed directly by introducing the random variables resulting from the model.

\begin{enumerate}
\item The number of photons emitted by the source is encoded as a random variable $N$ with outcomes $n\in\{0,\dots,\infty\}$ distributed according to the Poisson distribution 
\[P_N(n)=e^{-|\alpha|^2}\frac{|\alpha|^{2n}}{n!}.\]
\item The sensitivity of the detectors corresponds to the minimum number of photons that is needed for the detector to fire. This is modelled for each detector $D_{v,h}$ by a random variable $R_{v,h}$ with outcomes $r_{v,h}\in\{1,\dots n\}$. The distribution is given by
\begin{equation}\label{PR}
P_{R_{v,h}}(r_{v,h})=\mu(1-\mu)^{r_{v,h}-1}.
\end{equation}
The detector $D_{v,h}$ clicks if at least $r_{v,h}$ photons arrive. Eq.~\eqref{PR} is the probability that the detector did not fire for the for  first $r_{v,h}-1$ photons and that it a click is induced by the $r_{v,h}^{th}$ photon. Then, the probability that $m$ incoming photons are detected is equal to
\[\sum_{r_{v,h}=1}^mP_{R_{v,h}}(r_{v,h})=1-(1-\mu)^m,\]
which can be found using a geometric series. This is equal to one minus the probability that none of the photons is detected, which is what we expect.
\item Dark counts, afterpulses and crosstalk correspond to the side information whether a detector fires independently of whether photons arrive at it or not. This is encoded for each detector by random variable $S_{v,h}$ with outcomes $s_{v,h}\in\{0,1\}$, where $s_{v,h}=1$ corresponds to a such a deterministic click. The distribution is given by
\[P_{S_{v,h}}(s_{v,h}=1)=1-(1-p_{dark})\cdot(1-p_{\gamma})\cdot(1-p_{\delta}),\]
where $p_{dark}$ is the probability to have a dark count, $p_{\gamma}$ is the probability for after pulses and $p_{\delta}$ is the crosstalk-probability. If $X_{v,h}^i$ is the bit generated in the $i$-th run, then $p_{\delta}=\delta\cdot \Pr(x_{v,h}^{i-1}=1)$, where $\delta$ is a device-dependent parameter. Analogously we have $p_{\gamma}=\gamma \cdot P_X(x_{v,h}=1)$. If we assume that the probability distribution $P_X$ of the raw randomness is constant for each run we can omit the superscripts and simply write  ${\Pr(x_{v,h}^{i}=1)=\Pr(x_{v,h}=1)}$. If we also take it to be symmetric, we can define $p_x:=P_{X_{v,h}}(x_{v,h}=1)$, such that we have
\begin{equation}\label{ps}
P_{S_{v,h}}(s_{v,h}=1)=1-(1-p_{dark})\cdot(1-\gamma\cdot p_x)\cdot(1-\delta\cdot p_x).
\end{equation}
\end{enumerate}
The quantum randomness corresponds to a random variable $Q$ with uniformly distributed outcomes $q\in\{v,h\}$.\\
The final randomness is a function $\chi(Q^{\infty},N,S_v,S_h,R_v,R_h)=(x_v,x_h)$
\[x_{v}=\left\{\begin{array}{ll}1&\text{if } s_v=1 \text{ or if } |\{i:q_i=v\}_{1\leq i\leq n}|\geq r_v\\
0&\text{else}
\end{array}\right.\] 
\[x_{h}=\left\{\begin{array}{ll}1&\text{if } s_h=1 \text{ or if } |\{i:q_i=h\}_{1\leq i\leq n}|\geq r_h\\
0&\text{else}
\end{array}\right.\] 
To calculate $P_{S_{v,h}}$ explicitly, we first need to determine $p_x$, which can be done recursively using Eq.~\eqref{ps}\footnote{We write $R=R_{v,h}$ and $S=S_{v,h}$.}
\begin{align}
p_x&=P_S(s=1)+P_{S}(s=0)\underbrace{\left(\sum_{n=1}^\infty P_N(n)\sum_{r=1}^nP_{R}(r)\sum_{m=r}^n\left(\frac{1}{2}\right)^n\left(\begin{array}{c}n\\m\end{array}\right)\right)}_{:=p_{det}}\nonumber\\[0.3cm]
&=(1-p_{dark})\cdot(1-\gamma\cdot p_x)\cdot(1-\delta\cdot p_x)(1-p_{det})+p_{det}\label{px}
\end{align}
This can be solved for $p_x$ and reinserted into $P_{S_{v,h}}$ \eqref{ps}.
\begin{center}
\begin{table}[h]
{\renewcommand{\arraystretch}{1.5}
\renewcommand{\tabcolsep}{0.3cm}
\resizebox{\columnwidth}{!}{
\begin{tabular}{|c|c|c|c|c|}
\hline
$$&\multicolumn{4}{|c|}{$p(x_vx_h|nr_vr_h)=P_{X|NR_vR_h}(x_vx_h|nr_vr_h)$}\\
\hline
\multicolumn{5}{c}{$$}\\[-0.5cm]
\hline
$s_vs_h=(1,1)$&\multicolumn{4}{|c|}{$$}\\
\hline
&$p(00|nr_vr_h)$&$p(01|nr_vr_h)$&$p(10|nr_vr_h)$&$p(11|nr_vr_h)$\\[0.1cm]
\hline
&$0$&$0$&$0$&$1$\\
\hline
\multicolumn{5}{c}{$$}\\[-0.5cm]
\hline
$s_vs_h=(0,1)$&\multicolumn{4}{|c|}{$$}\\
\hline
&$p(00|nr_vr_h)$&$p(01|nr_vr_h)$&$p(10|nr_vr_h)$&$p(11|nr_vr_h)$\\[0.1cm]
\hline
&$0$&$\left(\frac{1}{2}\right)^n\sum\limits_{m=0}^{r_h-1}\left(\begin{array}{c}n\\m\end{array}\right)$&$0$&$\left(\frac{1}{2}\right)^n\sum\limits_{m=r_h}^{n}\left(\begin{array}{c}n\\m\end{array}\right)$\\
\hline
\multicolumn{5}{c}{$$}\\[-0.5cm]
\hline
$s_vs_h=(1,0)$&\multicolumn{4}{|c|}{$$}\\
\hline
&$p(00|nr_vr_h)$&$p(01|nr_vr_h)$&$p(10|nr_vr_h)$&$p(11|nr_vr_h)$\\[0.1cm]
\hline
&$0$&$0$&$\left(\frac{1}{2}\right)^n\sum\limits_{m=0}^{r_v-1}\left(\begin{array}{c}n\\m\end{array}\right)$&$\left(\frac{1}{2}\right)^n\sum\limits_{m=r_v}^{n}\left(\begin{array}{c}n\\m\end{array}\right)$\\
\hline
\multicolumn{5}{c}{$$}\\[-0.5cm]
\hline
$s_vs_h=(0,0)$&\multicolumn{4}{|c|}{$$}\\
\hline
&$p(00|nr_vr_h)$&$p(01|nr_vr_h)$&$p(10|nr_vr_h)$&$p(11|nr_vr_h)$\\[0.1cm]
\hline
$r_v,r_h>n$&$1$&$0$&$0$&$0$\\
\hline
$r_h\leq n,r_v>n$&$\left(\frac{1}{2}\right)^n\sum\limits_{m=0}^{r_h-1}\left(\begin{array}{c}n\\m\end{array}\right)$&$\left(\frac{1}{2}\right)^n\sum\limits_{m=r_h}^{n}\left(\begin{array}{c}n\\m\end{array}\right)$&$0$&$0$\\
\hline
$r_h>n, r_v\leq n$&$\left(\frac{1}{2}\right)^n\sum\limits_{m=0}^{r_v-1}\left(\begin{array}{c}n\\m\end{array}\right)$&$0$&$\left(\frac{1}{2}\right)^n\sum\limits_{m=r_v}^{n}\left(\begin{array}{c}n\\m\end{array}\right)$&$0$\\
\hline
$\begin{array}{c}r_v,r_h\leq n\\ r_h+r_v\leq n\\
\begin{tikzpicture}
\tiny
\draw (0,0)--(2.5,0);
\draw (0.6,-0.05)--(0.6,0.05);
\draw (1.9,-0.05)--(1.9,0.05);
\draw (0,-0.05)--(0,0.05);
\draw (2.5,-0.05)--(2.5,0.05);
\node at (0.3,0.2) {01};
\node at (1.25,0.2) {11};
\node at (2.2,0.2) {10};
\node at (0.6,-0.2) {$r_v$};
\node at (1.9,-0.2) {$n-r_h$};
\node at (2.7,0) {$m$};
\normalfont
\end{tikzpicture}
\end{array} $
&$0$&$\left(\frac{1}{2}\right)^n\sum\limits_{m=0}^{r_v-1}\left(\begin{array}{c}n\\m\end{array}\right)$&$\left(\frac{1}{2}\right)^n\sum\limits_{m=n-r_h+1}^{n}\left(\begin{array}{c}n\\m\end{array}\right)$&$\left(\frac{1}{2}\right)^n\sum\limits_{m=r_v}^{n-r_h}\left(\begin{array}{c}n\\m\end{array}\right)$\\
\hline
$\begin{array}{c}r_v,r_h\leq n\\ r_h+r_v> n\\
\begin{tikzpicture}
\tiny
\draw (0,0)--(2.5,0);
\draw (0.6,-0.05)--(0.6,0.05);
\draw (1.9,-0.05)--(1.9,0.05);
\draw (0,-0.05)--(0,0.05);
\draw (2.5,-0.05)--(2.5,0.05);
\node at (0.3,0.2) {01};
\node at (1.25,0.2) {00};
\node at (2.2,0.2) {10};
\node at (0.6,-0.2) {$n-r_h$};
\node at (1.9,-0.2) {$r_v$};
\node at (2.7,0) {$m$};
\normalfont
\end{tikzpicture}
\end{array} $&$\left(\frac{1}{2}\right)^n\sum\limits_{m=n-r_h+1}^{r_v-1}\left(\begin{array}{c}n\\m\end{array}\right)$&$\left(\frac{1}{2}\right)^n\sum\limits_{m=0}^{n-r_h}\left(\begin{array}{c}n\\m\end{array}\right)$&$\left(\frac{1}{2}\right)^n\sum\limits_{m=r_v}^{n}\left(\begin{array}{c}n\\m\end{array}\right)$&$0$\\
\hline
\end{tabular}
}
}
\caption{\emph{Raw randomness conditioned on side information.} Probability distribution of the QRNG output $X$ conditioned on side information $S_vS_h$, $R_vR_h$ and $N$.}
\label{distr2}
\end{table}
\end{center}
The joint distribution $P_{S_vS_h}(s_vs_h)$ is in general not equal to the product distribution $P_{S_v}(s_v)P_{S_h}(s_h)$. For the calculation of $H_{\min}(X|NS_vS_hR_vR_h)$ we minimise over all $P_{S_vS_h}(s_vs_h,y)$ subject to the constraint $P_{S_{v,h}}(s_{v,h}=1):=p$. The free parameter is $0\leq y\leq p$.\\

The conditional min-entropy is then equal to
\begin{multline*}
H_{\min}(X|NS_vS_hR_vR_h)\\=\min_y-\log\left[P_N(0)+\sum\limits_{n=1}^{\infty}P_N(n)\left(\sum\limits_{s_v,s_h,r_v,r_h}P_{R_v}(r_v)P_{R_h}(r_h)P_{S_vS_h}(s_vs_h,y)\max\limits_{x_vx_h}P_{X|S_vS_hR_vR_h}(x_vx_h|s_vs_hr_vr_h)\right)\right],
\end{multline*}
where the distribution of $P_{R_{v,h}}(r_{v,h})$ and $P_{S_vS_h}(s_v,s_h)$ are given by \eqref{PR} and \eqref{ps} respectively. The guessing probabilities can be found in Table~\ref{distr2}. The resulting min-entropy is shown in Fig.~\ref{HXSRN}.
\begin{figure}
\includegraphics[scale=0.6]{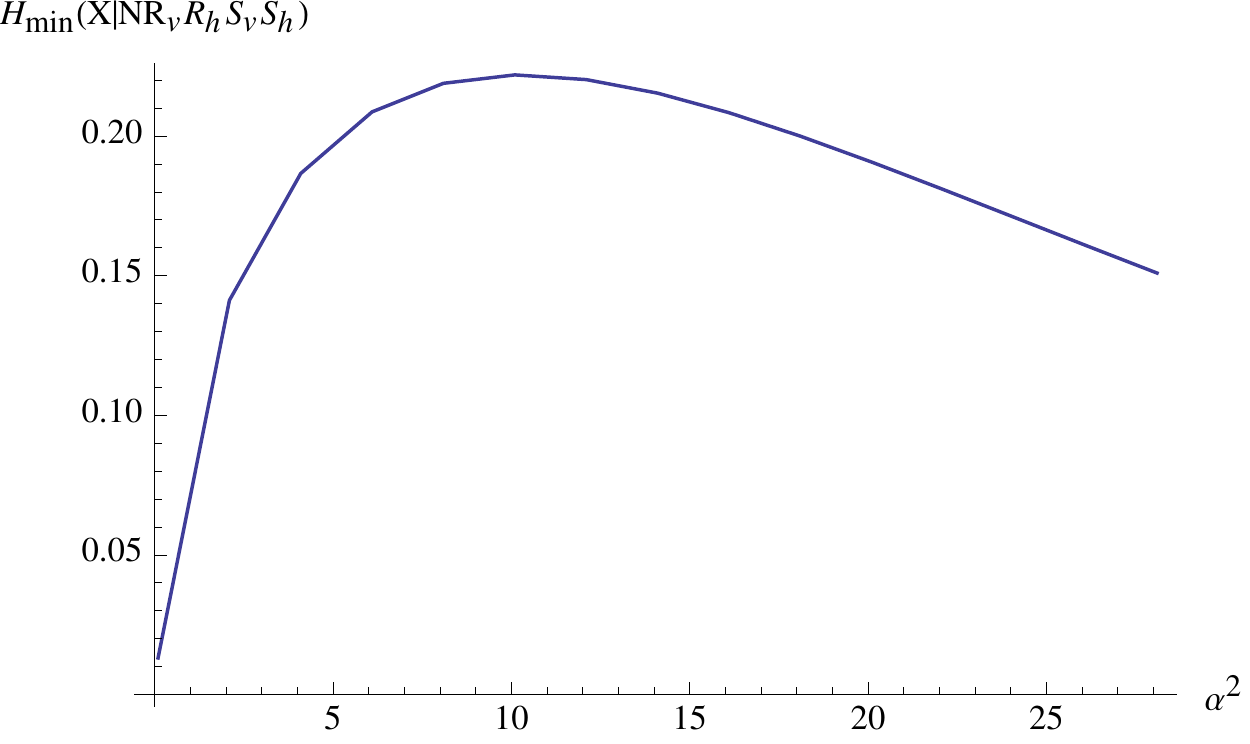}
\includegraphics[scale=0.6]{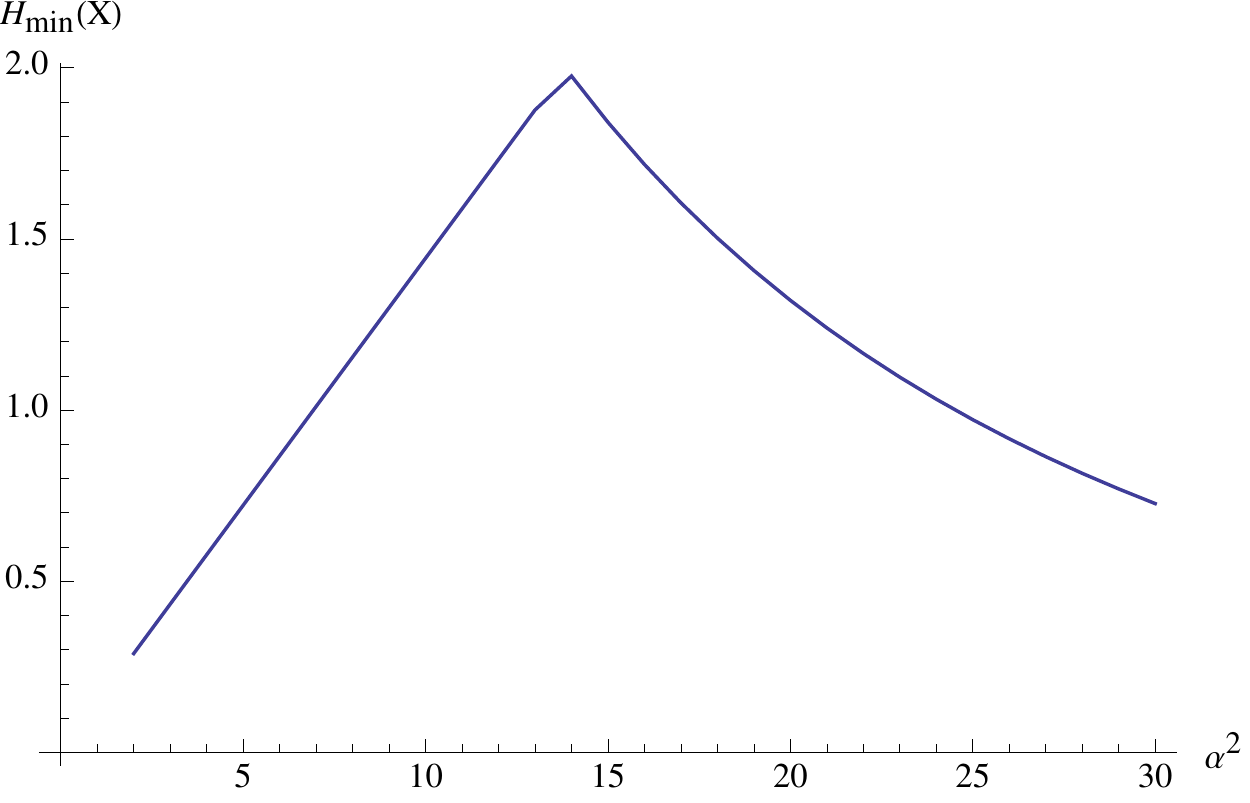}
\caption{Left: Extractable bit rate for a PBS-based QRNG such that the resulting randomness is independent of side information due to the source, limited detector efficiencies, afterpulses, cross talk and dark counts. Right: Extractable uniformly distributed bit rate: The resulting randomness may still depend on side information. The parameters for both plots are $\mu=0.1$, $p_{dark}=10^{-6}$, $\gamma=\delta=10^{-3}$.}\label{HXSRN}
\end{figure}
The extractable uniformly distributed bit rate can be found using Eq.~\eqref{px}
\[
H_{\min}(X)=-\log\left[\max_x P_X(x)\right].
\]
\subsection*{Additional Randomness from Arrival Time}
One possibility to increase the bit rate of the QRNG is to consider additional timing information. For example, if light pulses are sent out at times $n T$ (with $n \in \mathbb{N}$) one may add to the raw randomness for each detector the information whether the click was noticed in the time interval $[n T- T/2,  n T]$ or in $[n T,n T + T/2 ]$. 
To calculate the extractable randomness of such a modified scheme, one can apply the above analysis with pulses of half the original length. The idea is that one pulse can be seen as two pulses of half the original intensity (which is true if we assume a Poisson distribution of the photon number in the pulses). If there were no correlations between the pulses this would lead to a doubling of the bit rate (for appropriately chosen intensity of the source). However, because of the limited speed of the detectors (dead time and afterpulses) correlations between pulses may increase drastically when operating the device at a higher speed, which may again reduce the bit rate.

\section{Efficient Implementation of Randomness Extraction}
\label{sec:impl}

\begin{figure}
\lstset{language=C}
\begin{lstlisting}
#include <stdint.h>

const unsigned n=1024; // CHANGE to the number of input bits
const unsigned l=768;  // CHANGE to the number of output bits


// the extraction function
// parameters:
//   y: an output array of l bits stored as l/64 64-bit integers
//   m: a random matrix of l*n bits, stored in l*n/64 64-bit integers
//   x: an input array of n bits stores as n/64 64-bit integers

void extract(uint64_t * y, uint64_t const * m, uint64_t const * x) 
{
  assert (n%64==0 && l%64 == 0);
  
  int ind =0;
  // perform a matrix-vector multiplication by looping over all rows
  // the outer loop over all words
  for (int i = 0; i < l/64; ++i) {
    y[i]=0;
    // the inner loop over all bits in the word
    for (unsigned j = 0; j < 64; ++j) {
      uint64_t parity = m[ind++] & x[0];
      // perform a vector-vector multiplication using bit operations
      for (unsigned l = 1; l < n/64; ++l)
        parity ^= m[ind++] & x[l];
      // finally obtain the bit parity
      parity ^= parity >> 1;
      parity ^= parity >> 2;
      parity = (parity & 0x1111111111111111UL) * 0x1111111111111111UL;
      // and set the j-th output bit of the i-th output word
      y[i] |= ((parity >> 60) & 1) << j;
    }
  }
}

\end{lstlisting}
\caption{\emph{Efficient implementation of two-universal hashing on a 64-bit CPU in the C99 programming language.} The parameters $n$ and $\ell$ have to be multiples of 64 and can be changed in the second and third row of the code respectively.}
\label{fig:code}
\end{figure}
In this section we present an efficient implementation of randomness extraction by two-universal hashing. Given a random $\ell\times n$ bit matrix $m_{ij}$ two-universal hashing $Y = f(X)$ requires the evaluation of the matrix-vector product
\begin{equation}
Y_i =\sum_{j=1}^n m_{ij}X_j 
\end{equation}
to be performed modulo 2. This can be done very efficiently on modern CPUs using bit operations. Storing 32  (64) entries of the vector $X$ in a 32-bit (64-bit) integer, multiplication is implemented by bitwise $\mathsf{AND}$ operations and addition modulo 2 by bitwise $\mathsf{XOR}$ operations. A sum modulo 2 over all entries maps to the bit parity of the integer. An efficient implementation of two-universal hashing is given in Fig. \ref{fig:code}. The source code shown in this figure and optimised versions using explicitly vectorised compiler intrinsics for SSE4.2 are provided as supplementary material.

\end{appendix}
\end{document}